\documentclass[12pt,fleqn]{article}
\usepackage{mathtools}
\usepackage[margin=0.7in]{geometry}
\usepackage{amsmath}
\usepackage{amsfonts}
\usepackage{amsthm}
\usepackage{comment}
\usepackage{multirow}
\usepackage{graphicx}
\usepackage{dsfont}
\usepackage{tikz}
\usepackage{lscape}
\usetikzlibrary{decorations}
\usetikzlibrary{decorations.markings,arrows}

\newtheorem{Thm}{Theorem}[section]

\newtheorem{Prop}[Thm]{Proposition}

\theoremstyle{definition}

    \def\Re{{\rm Re \,}}
    \def\Im{{\rm Im \,}}

    \def\P2n{{\rm P}_{{\rm II}}^{(n)}}

    \newtheorem{theorem}{Theorem}[section]
    \newtheorem{lemma}[theorem]{Lemma}
    \newtheorem{corollary}[theorem]{Corollary}
    \newtheorem{proposition}[theorem]{Proposition}
    
    \newtheorem{Definition}[theorem]{Definition}
    
    \newtheorem{Remark}[theorem]{Remark}
    
    \newtheorem{Example}[theorem]{Example}
    
    \newtheorem{Assumptions}[theorem]{Assumptions}

\begin{document}
\title{Uniform asymptotics of Toeplitz determinants with Fisher--Hartwig singularities} \author{B. Fahs \footnote{Department of Mathematics, Imperial College London, London, UK. Email: bf308@ic.ac.uk. } }
\maketitle
\begin{abstract}
We obtain an asymptotic formula for  $n\times n$ Toeplitz determinants as $n\to \infty$, for real valued symbols with  any fixed number of  Fisher--Hartwig singularities, which is uniform with respect to the location of the singularities. 

As an application, we prove a conjecture by Fyodorov and Keating \cite{FK} regarding moments of averages of the characteristic polynomial  of the Circular Unitary Ensemble.

In addition,  we  obtain an asymptotic formula regarding the momentum of impenetrable bosons in one dimension with periodic boundary conditions.

\end{abstract}
\section{Introduction}

\label{SectionToep}
In this paper, we consider the asymptotics as $n\to \infty$ of Toeplitz determinants 
\begin{equation}\nonumber
D_n(f)=\det\left(f_{j-k}\right)_{j,k=0}^{n-1},\qquad  f_j=\int_0^{2\pi}f\left(e^{i\theta}\right)e^{-ij\theta }\frac{d\theta}{2\pi},
\end{equation}
where the symbol $f$ is of the form 
\begin{equation}\label{symbol}
f(z)=e^{V(z)}\omega(z), \qquad \omega(z)=\prod_{j=1}^m  \omega_{\alpha_j,\beta_j}(z/z_j), \qquad
\omega_{\alpha,\beta}(z)=\left(\frac{z}{e^{\pi i}}\right)^{\beta}\left|z-1\right|^{2\alpha},
\end{equation}
satisfying the following conditions:
\begin{itemize}
 \item[(a)] $V(z)$ is real--valued for $|z|=1$ and is analytic on an open set containing $|z|=1$,
\item[(b)] $z_j=e^{it_j}$, where $0\leq t_1<t_2<\dots<t_m<2\pi$,
\item[(c)] $\alpha_j\geq0 $ and $\Re \beta_j = 0$ for $j=1,2,\dots,m$.\end{itemize}
Under these conditions $f$ is a real valued symbol, and
we obtain large $n$ asymptotics of $D_n(f)$ (up to a bounded multiplicative term) which are uniform in the parameters $z_j$.

When $(\alpha_j,\beta_j)\neq (0,0)$ for all $j$,  one says the Toeplitz determinant possesses a Fisher--Hartwig (FH) singularity at each point $z_j=e^{it_j}$, and that the singularity at $z_j$ is of root--type if $\beta_j=0$ and of jump--type if $\alpha_j=0$.  

The large $n$ asymptotics of Toeplitz determinants were first studied by Szeg\H{o} in 1915 \cite{Szego15}. They have been intensively studied over the last 70 years, and owe their relevance  to applications to physical models. The most prominent such application is the question of spontaneous magnetization of the Ising model on the lattice $\mathbb Z^2$ (see e.g. \cite{McCoyWu} and  \cite{DIK2}), but we also mention questions surrounding the momentum of impenetrable bosons in 1 dimension, which we return to in Section \ref{SecBos}. 

In addition to physical models, a considerable effort has been invested in understanding statistical similarities between the asymptotics of the Riemann zeta function along the critical line $\Re z=1/2$ and the statistics of the characteristic polynomial of the Circular Unitary Ensemble (CUE) over arcs of the unit circle. Toeplitz determinants appear in this context, and we return to this topic in Section \ref{SectionCUE}.

We now turn to known results for the asymptotics of $D_n(f)$. The simplest case is the special one where $\omega(z)\equiv 1$ (i.e. $\alpha_j,\beta_j=0$ for all $j$), in which case 
\begin{equation}\nonumber
D_n(e^V)=e^{nV_0} e^{\sum_{k=1}^\infty kV_kV_{-k}}(1+o(1)), \end{equation}
as $n\to \infty$, where $V_k=\int_0^{2\pi}V\left(e^{i\theta}\right)e^{-ik\theta }\frac{d\theta}{2\pi}
$.
This is known as the  strong Szeg\H{o} limit theorem (see \cite{Golinskii,Ibragimov,Johansson,Szego}), and holds for $V$ satisfying condition (a), but also more generally for any $V$ such that $\sum_{k=-\infty}^\infty |V_k|^2$ converges. 

It was conjectured by Lenard \cite{Lenard2} Fisher and Hartwig \cite{FH}, and proven in subsequent steps by Widom \cite{Widom} (relying also on work by Lenard \cite{Lenard2}) and Basor \cite{Basor,Basor2},  that if $f$ is a symbol of the form \eqref{symbol} satisfying (a)--(c), then
\begin{equation}\label{formulaWidom}
D_n(f)=E n^{\sum_{j=1}^m\alpha_j^2-\beta_j^2}e^{nV_0}(1+o(1)),
\end{equation}
as $n\to \infty$, where $E$ is independent of $n$ and given by
\begin{multline}\label{FormulaE}
E
= e^{\sum_{k=1}^\infty kV_kV_{-k}}
\prod_{j=1}^m\left( \frac{G(1+\alpha_j+\beta_j)G(1+\alpha_j-\beta_j)}{G(1+2\alpha_j)}
e^{2\Re\left[(\beta_j-\alpha_j)V_+(z_j)\right]}
\right)
\\ \prod_{1\leq j<k\leq m}\left(|e^{it_j}-e^{it_k}|^{2(\beta_j\beta_k-\alpha_j\alpha_k)}e^{i(t_k-t_j-\pi)(\alpha_j\beta_k-\alpha_k\beta_j)}
\right),
\end{multline}
  $V_+(z)=\sum_{j=1}^\infty V_{ j}z^{ j}$, and $G(z)$ is the Barnes' $G$-function (see e.g. \cite{NIST}). 
  We  mention here that although our focus is on real--valued symbols, the analogue of \eqref{formulaWidom}, \eqref{FormulaE} in the case of complex symbols $f$ is interesting and exhibits behaviour with additional subtleties, see \cite{BottSil,DIK,DIK3}, or \cite{DIK2} for a review.

By the proof in  \cite{DIK3}, it is clear that the asymptotics \eqref{formulaWidom} hold uniformly for $e^{it_j}$ and $e^{it_k}$ bounded away from each other.
It is also clear that the asymptotics \eqref{formulaWidom} are discontinuous if any two points $e^{it_j}, e^{it_k}$ merge and that the asymptotic formula cannot be correct in this situation. In \cite{JMMS} and \cite{VadTr}, for  $\alpha=1/2$ and $\beta=0$,  a part of the transition was considered, corresponding to the box $|e^{it_j}-1|<C/n$ for all $j=1,\dots,m$ and some fixed, large constant $C$. More recently, in \cite{CK}, the authors considered the situation where $m=2$ and obtained the full asymptotics of $D_n(f)$, uniformly for $0\leq t_1<t_2<2\pi$. It is easily seen that the results of \cite{CK} may be presented in the following manner:
\begin{equation}\label{formulaCK}
\log D_n(f)=nV_0+\sum_{j=1}^2\left(\alpha_j^2-\beta_j^2\right)\log n+2(\alpha_1\alpha_2-\beta_1\beta_2)\log \left(\frac{1}{\sin\left|\frac{t_1-t_2}{2}\right|+n^{-1}}\right)+\widehat F_n(t_1,t_2)+o(1),
\end{equation}
uniformly as $n\to \infty$, 
where $\widehat F_n$ is an explicit function in $t_1,t_2$ which is uniformly bounded as $n\to \infty$. We mention that $\widehat F_n $ has an interesting and intricate representation involving a solution to the Painlev\'e V equation when $|t_1-t_2|=\mathcal O(1/n)$ -- for the details we refer the reader to \cite{CK} (and additionally to \cite{CF} for certain simplifications that occur in the specific case where $\alpha$ is integer--valued and $\beta=0$).

We also refer the reader to work on a different but related problem, namely the transition between smooth symbols and those with one singularity, see \cite{CIK} and \cite{WMTB}.

In this paper, we obtain asymptotics for $D_n(f)$ as $n\to \infty$, uniformly in the parameters $t_1,\dots,t_m$. Our main result is the following.

\begin{theorem}\label{Theorem1}
Assume that $f$ is of the form \eqref{symbol}, satisfying (a) -- (c). Then as $n\to \infty$,
\begin{equation}\nonumber
\log D_n(f)=nV_0+\sum_{j=1}^m\left(\alpha_j^2-\beta_j^2\right)\log n+\sum_{1\leq j< k\leq m}2(\alpha_j\alpha_k-\beta_j\beta_k)\log \left(\frac{1}{\sin\left|\frac{t_j-t_k}{2}\right|+n^{-1}}\right)+\mathcal O(1),
\end{equation}
where the error term is uniform for $0\leq t_1<t_2<\dots<t_m<2\pi$.
\end{theorem}

\section{The characteristic polynomial of the CUE}
 \label{SectionCUE}
Let $Z_1,Z_2,\dots,Z_n$ be random variables, distributed as the eigenvalues of the $n\times n$ Circular Unitary Ensemble of random matrices, with the following joint probability density function on the unit circle in the complex plane:
\begin{equation}\label{ProbDens}
\frac{1}{n!}\prod_{1\leq i <j\leq n}|z_i-z_j|^2\prod_{j=1}^n \frac{dz_j}{2\pi i z_j}.
\end{equation}
Let the characteristic polynomial be denoted by
\begin{equation}\nonumber
P_n(e^{i\theta})=\prod_{j=1}^n \left(Z_j-e^{i\theta}\right).
\end{equation}
It has long been believed that the statistical properties of the Riemann zeta function on the critical line $s=1/2+it,$ $t\in \mathbb R$, and the statistics of large random matrices are related -- it was Dyson who first spotted this possible connection. More recently, possible connections between the behaviour of the characteristic polynomial $P_n(e^{i\theta})$ over the unit circle and the behaviour of the Riemann zeta function along the critical line have been studied intensively (see e.g. \cite{CFKRS1,CFKRS,CFZ,FK,FHK,GHK,HKC,KS,Mont} and references therein).
 In this context the authors of \cite{FHK,FK} were interested in both extreme values and average values  of $P_n(e^{it})$ over the unit circle, namely the random variables
\begin{equation}\nonumber
Y_n=\max_{t\in[0,2\pi)}\left|P_n\left(e^{it}\right)\right|,\qquad
X_n(\alpha)=\int_0^{2\pi} \left|P_n(e^{it})\right|^{2\alpha}\frac{dt}{2\pi}.
\end{equation}
In addition, their work sparked interest in connections between the characteristic polynomial of the CUE and Gaussian Multiplicative Chaos, see \cite{Webb1,Webb2} for results on such connections.

In \cite{FHK} it was conjectured that $\log Y_n-\log n+\frac{3}{4}\log \log n$ converges in distribution to a random variable, and subsequently the asymptotics  of $Y_n$ have been studied in \cite{ABB,PZ,CMN}. In these works, the terms $\log n$ and $\frac{3}{4}\log \log n$ were confirmed. The full conjecture, however, remains open.

 In \cite{FK}, Fyodorov and Keating conjectured that for $m=1,2,\dots$,
\begin{equation}\label{conjecture}
\log \mathbb E\left[X_n(\alpha)^m\right]=
\begin{cases}
m\alpha^2 \log n + C_m(\alpha)+o(1)&{\rm for }\, \, 0<m\alpha^2<1,\\
\left[(m\alpha)^2+1-m\right]\log n+\mathcal O(1)&{\rm for }\, \, m\alpha^2>1,
\end{cases}
\end{equation} 
as $n\to \infty$, where $\mathbb E$ denotes the expectation with respect to \eqref{ProbDens}, and
\begin{equation} \nonumber
C_m(\alpha)=\log \left[\left(\frac{G(1+\alpha)^2}{G(1+2\alpha)\Gamma(1-\alpha^2)}\right)^m\Gamma(1-m\alpha^2)\right]. \end{equation}
 As a corollary to Theorem \ref{Theorem1}, we will prove \eqref{conjecture}.

To make the connection between Toeplitz determinants and the moments of $X_n(\alpha)$, we recall the well--known representation of Toeplitz determinants in terms of multiple integrals
\begin{equation}\label{intrep}
D_n(f)=\frac{1}{n!} \int_{[0,2\pi)^n}\prod_{1\leq j<k\leq n}\left|e^{ij\theta}-e^{ik\theta}\right|^2\prod_{j=1}^nf\left(e^{i\theta_j}\right)\frac{d\theta_j}{2\pi},
\end{equation}  
 from which it follows that 
\begin{equation}\label{XnToep}
\mathbb E\left[X_n(\alpha)^m\right]=\int_{[0,2\pi)^m}D_n\left(f_m^{(\alpha)}\right)\frac{dt_1}{2\pi}\dots \frac{dt_m}{2\pi},
\end{equation}
where we denote
\begin{equation}\label{falpha}
f_{m}^{(\alpha)}(z)=\prod_{j=1}^{m}\left|z-e^{it_j}\right|^{2\alpha}.
\end{equation}

Using \eqref{XnToep} and \eqref{formulaCK}, Claeys and Krasovsky were able to prove \eqref{conjecture} for $m=2$. They were furthermore able to prove that $\mathbb E\left[X_n(\alpha=1/\sqrt{2})^2\right]=\hat c n\log n(1+o(1))$ as $n\to \infty$ for an explicit constant $\hat c$. Additionally, for $2\alpha^2>1$, they were able to determine explicitly the $\mathcal O(1)$ term in \eqref{conjecture} in terms of the Painlev\'{e} V equation.

 For integer $m\geq 2$, the conjecture was proven recently by Bailey and Keating \cite{BK} for integer $\alpha=1,2,3,\dots$, by representing $\mathbb E\left[X_n(\alpha)^m\right]$ in terms of integrals of Schur polynomials. A second proof was given by Assiotis and Keating in \cite{AK}, where a representation for the constant $\mathcal O(1)$ term was given in terms of a certain volume of continuous Gelfand--Tsetlin patterns with constraints.
 
Relying on Theorem \ref{Theorem1} and \eqref{formulaWidom}, we prove the conjecture \eqref{conjecture} for all parameter sets, and more precisely we have following corollary.

\begin{corollary}\label{Corollary}
As $n\to \infty$, the asymptotics of $\log \mathbb E\left[X_n(\alpha)^m\right]$ are given by \eqref{conjecture}, both for $0<m\alpha^2<1$ and for $m\alpha^2>1$. Additionally, if $m\alpha^2=1$, then 
\begin{equation}\nonumber
\log \mathbb E\left[X_n(\alpha)^m\right]=\log n+\log \log n+\mathcal O(1),
\end{equation}
as $n\to\infty$.
\end{corollary}

\subsection{Proof of Corollary \ref{Corollary}}
\subsubsection{$0<m\alpha^2<1$}\label{malp<1}
Denote 
\begin{equation}\label{defIalpha}
I_{\epsilon}(\alpha)=\int_{[0,2\pi)^m}\prod_{1\leq j<k\leq m} \left(\sin\left|\frac{t_j-t_k}{2}\right|+\epsilon\right)^{-2\alpha^2} dt_1\dots dt_m,
\end{equation}
 for $\epsilon>0$. If $0<m\alpha^2<1$, the integral is well defined for $\epsilon=0$ and $I_0(\alpha)$ is a Selberg integral \cite{Selberg,FW}:
\begin{equation}\label{Selberg}
I_{0}(\alpha)=(2\pi)^m\frac{\Gamma(1-m\alpha^2)}{\Gamma(1-\alpha^2)^m}.
\end{equation}

Thus Corollary \ref{Corollary} follows (for $0<m\alpha^2<1$) by combining \eqref{formulaWidom}, \eqref{XnToep} and Theorem \ref{Theorem1} as follows. Fix $\delta>0$. We will show that there is an integer $N$ such that for $n>N$,
\begin{equation}\label{interCor}
\left|\mathbb E\left[X_n(\alpha)^m\right]-n^{m\alpha^2}e^{C_m(\alpha)}\right|<\delta n^{m\alpha^2},\end{equation}
thus proving the corollary. Given a measurable subset $R\subset [0,2\pi)^m$, we denote
\begin{equation}
I_{\epsilon}(\alpha, R)=\int_{R} \,\,\prod_{1\leq j<k\leq m} \left(\sin\left|\frac{t_j-t_k}{2}\right|+\epsilon\right)^{-2\alpha^2} dt_1\dots dt_m.
\end{equation}
We note that $I_{\epsilon}(\alpha,R)<I_0(\alpha,R)$ for any $\epsilon>0$ for any $R\subset[0,2\pi)^m$.
For $\eta>0$,
divide the integration regime $[0,2\pi)^m$ into two regions $R_1(\eta)$ and $R_2(\eta)$, where $R_1(\eta)$ is the region where $\sin\frac{|t_i-t_j|}{2}>\eta$ for all $i\neq j$, and $R_2(\eta)$ is the complement of $R_1(\eta)$. 
 It follows by Theorem \ref{Theorem1} that
 \begin{equation}\nonumber
\int_{R_2(\eta)}D_n\left(f_m^{(\alpha)}\right)\frac{dt_1}{2\pi}\dots \frac{dt_m}{2\pi}=\mathcal O\left(n^{m\alpha^2}I_0(\alpha,R_2(\eta))\right),
\end{equation}
as $n\to \infty$, uniformly for $0<\eta<\eta_0$. In particular, since $I_0(\alpha,R_2(\eta))\to 0$ as $\eta\to 0$, it follows that there exists $\eta_0>0$ and $N_0\in \mathbb N$ such that 
\begin{equation}\label{R2statement}
\int_{R_2(\eta)}D_n\left(f_m^{(\alpha)}\right)\frac{dt_1}{2\pi}\dots \frac{dt_m}{2\pi}< \delta n^{m\alpha^2}/2,\end{equation}
for $n>N_0$ and $\eta<\eta_0$, which gives the desired bound for the integral over $R_2(\eta)$. We now evaluate the integral over $R_1(\eta)$. By \eqref{formulaWidom}, it follows that
\begin{equation}\nonumber
\int_{R_1(\eta)}
D_n\left(f_m^{(\alpha)}\right)\frac{dt_1}{2\pi}\dots \frac{dt_m}{2\pi}=n^{m\alpha^2}\frac{G(1+\alpha)^{2m}}{G(1+2\alpha)^m}\int_{R_1(\eta)}\prod_{1\leq j<k\leq m}|e^{it_j}-e^{it_k}|^{-2\alpha^2}(1+o(1))\frac{dt_1}{2\pi}\dots \frac{dt_m}{2\pi} ,\end{equation}
where the $o(1)$ tends to zero uniformly over $R_1(\eta)$ for any fixed $\eta$ as $n\to \infty$. Thus we may move the error term outside the integral, and since $R_2(\eta)$ is the complement of $R_1(\eta)$, we have
\begin{multline}
\int_{R_1(\eta)}\prod_{1\leq j<k\leq m}|e^{it_j}-e^{it_k}|^{-2\alpha^2}(1+o(1))\frac{dt_1}{2\pi}\dots \frac{dt_m}{2\pi}\\=\int_{[0,2\pi)^m}\prod_{1\leq j<k\leq m}|e^{it_j}-e^{it_k}|^{-2\alpha^2}\frac{dt_1}{2\pi}\dots \frac{dt_m}{2\pi}(1+o(1))\\-\int_{R_2(\eta)}\prod_{1\leq j<k\leq m}|e^{it_j}-e^{it_k}|^{-2\alpha^2}\frac{dt_1}{2\pi}\dots \frac{dt_m}{2\pi}.
\end{multline}
By \eqref{Selberg} we obtain
\begin{equation}\label{intR1}
\int_{R_1(\eta)}D_n\left(f_m^{(\alpha)}\right)\frac{dt_1}{2\pi}\dots \frac{dt_m}{2\pi}=n^{m\alpha^2}e^{C_m(\alpha)}(1+o(1))-n^{m\alpha^2} I_0(\alpha,R_2(\eta))\left(\frac{G(1+\alpha)^2}{2\pi G(1+2\alpha)}\right)^m,
\end{equation}
where the $o(1)$ tends to zero for any fixed $\eta$ as $n\to \infty$. Again we use the fact that $I_0(\alpha,R_2(\eta))\to 0$ as $\eta\to 0$, from which it follows that we may pick $\eta<\eta_0$ such that the second term on the right hand side of \eqref{intR1} is less than $\delta n^{m\alpha^2} /4$.

Thus we may fix $\eta<\eta_0$ such that
\begin{equation}\nonumber
\left| \int_{R_1(\eta)}
 D_n\left(f_m^{(\alpha)}\right)\frac{dt_1}{2\pi}\dots \frac{dt_m}{2\pi}-n^{m\alpha^2}e^{C_m(\alpha)}\right| <\delta n^{m\alpha^2} /4 +o(n^{m\alpha^2}),
\end{equation}
as $n\to \infty$, and combined with \eqref{R2statement}, this proves \eqref{interCor}.

\subsubsection{$m\alpha^2>1$}
We now study the asymptotics of $I_{\epsilon}(\alpha)$ for $m\alpha^2>1$. A lower bound for $I_{\epsilon}(\alpha)$ is easily obtained-- by integrating over the box $|t_i-t_j|<\epsilon$ for all $i\neq j$, it is easily seen that there is a constant $c$ such that
\begin{equation}\label{LBIeps}
c\epsilon^{m-1-m(m-1)\alpha^2}<I_{\epsilon}(\alpha), \end{equation}
for $0<\epsilon<\epsilon_0$.

To prove Corollary \ref{Corollary}, we need to obtain a corresponding upper bound for $I_\epsilon(\alpha)$. We choose to work with the following integral instead:
\begin{equation}\nonumber
\widehat I_\epsilon(\alpha)=  \int_{0<t_1<\dots<t_m<1}  \prod_{1\leq j<k\leq m}  \left(|t_j-t_k|+\epsilon\right)^{-2\alpha^2} dt_1\dots dt_m.
\end{equation}

\begin{lemma} 
There is a constant $0<c$ such that for all sufficiently small $\epsilon>0$,
\begin{equation}\nonumber
 I_{\epsilon}(\alpha)\leq c\widehat I_\epsilon(\alpha).
\end{equation}
 \end{lemma}

\begin{proof}
Denote $U_j=[0,2\pi(j-1)/(m+1))\cup[2\pi j/(m+1),2\pi)$ for $j=1,2,\dots,m+1$. Since there are $m$ points $t_1,\dots,t_m$ and $m+1$ sets $U_j$, it follows that there is always a $j$ such that $\{t_1,\dots,t_m\}\subset U_j$, thus
 \begin{equation}\nonumber
  I_\epsilon(\alpha)\leq \sum_{j=1}^{m+1}\int_{U_j^m}
  \prod_{1\leq j<k\leq m} \left(\sin\left|\frac{t_j-t_k}{2}\right|+\epsilon\right)^{-2\alpha^2} dt_1\dots dt_m.\end{equation}
  It follows that
\begin{equation}\nonumber
 I_\epsilon(\alpha)
\leq (m+1)\int_{[0,2\pi-2\pi/(m+1))^m}\, \prod_{1\leq j<k\leq m} \left(\sin\left|\frac{t_j-t_k}{2}\right|+\epsilon\right)^{-2\alpha^2} dt_1\dots dt_m.\end{equation}
Furthermore, for $0\leq x \leq \pi-\pi/(m+1)$, one has
\begin{equation}\nonumber
\frac{x}{\pi}\sin \frac{\pi}{m+1} \leq \sin x \leq x, \end{equation}
and it follows that
\begin{equation}\label{chintreg2}
I_\epsilon(\alpha)
\leq \frac{(m+1)\pi^{2\alpha}}{\left(\sin\frac{\pi}{m+1}\right)^{2\alpha}}\int_{[0,2\pi-2\pi/(m+1))^m}\, \prod_{1\leq j<k\leq m} \left(\left|\frac{t_j-t_k}{2}\right|+\epsilon\right)^{-2\alpha^2} dt_1\dots dt_m.
\end{equation}
The lemma follows easily from \eqref{chintreg2}.
\end{proof}

We now take the change of variables $s_j=t_{j+1}-t_j$ for $j=1,\dots,m-1$ and find that
\begin{equation}\nonumber
\widehat I_\epsilon(\alpha)= \int_0^1 dt_1 \int \prod_{1\leq j\leq k\leq m-1}  \left(\sum_{i=j}^k s_i+\epsilon\right)^{-2\alpha^2} ds_1\dots ds_{m-1},
\end{equation}
with integration taken over $ s_1,\dots, s_{m-1}\geq 0$ such that $ \sum_{j=1}^m s_j< 1-t_1$, from which it follows that
\begin{equation} \label{UBIeps}
 \widehat I_\epsilon (\alpha)\leq 
 I_\epsilon^{(2)}(\alpha),
\end{equation}
where 
\begin{equation}\label{intI4}
I_\epsilon^{(2)}(\alpha)=\int_{[0,1)^m} \, \prod_{1\leq j\leq k\leq m-1}  \left(\sum_{i=j}^k s_i+\epsilon\right)^{-2\alpha^2} ds_1\dots ds_{m-1} .
\end{equation}
 If $2\alpha^2>1$, then $I_\epsilon^{(2)}(\alpha)$ is straightforward to evaluate -- one simply notes that
\begin{equation}\label{uppersingle}
 \left(\sum_{i=j}^k s_{i}+\epsilon\right)^{-2\alpha^2} \leq  \left(s_{i}+\epsilon\right)^{-2\alpha^2}, \end{equation}
 for any $i=j,\dots,k$, and thus
 \begin{equation}\nonumber 
 I_\epsilon^{(2)}(\alpha)\leq \int_{[0,1)^m}\prod_{j=1}^{m-1}\left(s_j+\epsilon\right)^{-2j\alpha^2}ds_1\dots ds_{m-1}.
 \end{equation}
 Separating out the variables, it follows that
 \begin{equation}\label{2alpha2>1}
 I_{\epsilon}^{(2)}(\alpha)=\mathcal O\left(\epsilon^{(m-1)(1-m\alpha^2)}\right),
 \end{equation}
 as $\epsilon\to 0$, for $2\alpha^2>1$. However, if $\frac{1}{m}<\alpha^2<\frac{1}{2}$, this approach fails to yield \eqref{2alpha2>1}, and in fact yields a worse error term\footnote{As an example where the approach fails to provide optimal error terms, consider $m=3$ and $\alpha^2=2/5$. Then we obtain 
 \begin{equation}\nonumber
 I_\epsilon^{(\alpha)}\leq \int_0^1ds_1\int_0^1ds_2(s_1+\epsilon)^{-4/5}(s_2+\epsilon)^{-8/5}=\mathcal O\left(\epsilon^{-3/5}\right), 
 \end{equation} as $\epsilon\to 0$. However the optimal bound we are looking to obtain is of order $\epsilon^{(m-1)(1-m\alpha^2)}=\epsilon^{-2/5}$.
 }. To achieve the optimal error term \eqref{2alpha2>1} also in the case $\frac{1}{m}<\alpha^2<\frac{1}{2}$, we need to consider ordered integrals.

Since the integral \eqref{intI4} is taken over all possible orderings, it follows that
\begin{equation}\nonumber
I_\epsilon^{(2)}(\alpha)=\sum_{\sigma\in S_{m-1}}\int_{\mathcal W_\sigma} \, \prod_{1\leq j\leq  k\leq m-1}  \left(\sum_{i=j}^k s_{i}+\epsilon\right)^{-2\alpha^2} ds_1\dots ds_{m-1} ,
\end{equation}
where the integral is taken over
\begin{equation}\nonumber
\mathcal W_\sigma=\left \{ 0<s_{\sigma(m-1)}<s_{\sigma(m-2)}<\dots<s_{\sigma(1)}<1 \right \}.\end{equation}
  Let 
 \begin{equation}\nonumber
 V(\ell)=\{(j,k):\, 1\leq j \leq k \leq m-1 \textrm{ and } j=\sigma(\ell) \textrm{ or } k=\sigma(\ell)\}. \end{equation}
 Let $S(\ell)=V(\ell)\setminus \cup_{j=1}^{\ell-1}V(j)$. Then $S(1),\dots, S(m-1)$ are disjoint, and 
 \begin{equation}\nonumber
 \prod_{1\leq j\leq k\leq m-1}  \left(\sum_{i=j}^k s_i+\epsilon\right)^{-2\alpha^2}=\, \prod_{\ell=1}^{m-1}\, \prod_{(j,k)\in S(\ell)}  \left(\sum_{i=j}^k s_i+\epsilon\right)^{-2\alpha^2}.\end{equation}
By \eqref{uppersingle} it follows that
 \begin{equation}\nonumber
 \prod_{1\leq j\leq k\leq m-1}  \left(\sum_{i=j}^k s_i+\epsilon\right)^{-2\alpha^2}\leq\, \prod_{\ell=1}^{m-1}\, \prod_{(j,k)\in S(\ell)}  \left(s_{\sigma(\ell)}+\epsilon\right)^{-2\alpha^2}.\end{equation}
 It is easily seen that $|S(\ell)|=m-\ell$, and it follows that
 \begin{equation}\label{boundssjs}
 \prod_{1\leq j\leq k\leq m-1}  \left(\sum_{i=j}^k s_i+\epsilon\right)^{-2\alpha^2}\leq\, \prod_{\ell=1}^{m-1}\,  \left(s_{\sigma(\ell)}+\epsilon\right)^{-2\alpha^2(m-\ell)}.\end{equation}

By \eqref{boundssjs},
 it follows that
\begin{equation}\label{ints1m1}
\int_{\mathcal W_\sigma} \, \prod_{1\leq j \leq  k\leq m-1}  \left(\sum_{i=j}^k s_{i}+\epsilon\right)^{-2\alpha^2} ds_1\dots ds_{m-1} 
\leq \int_{\mathcal W_\sigma} \, \prod_{\ell=1}^{m-1}  \left( s_{\sigma(\ell)}+\epsilon\right)^{-2\alpha^2(m-\ell)} ds_{\sigma(m-1)}\dots ds_{\sigma(1)}.
\end{equation}
If $m=2$ and $m\alpha^2>1$, then the right hand side is of order $\epsilon^{-2\alpha^2+1}$, and we are done. We assume that $m>2$, and
 integrate in $s_{\sigma(m-1)}$ on the right hand side of \eqref{ints1m1}.  The power of $s_{\sigma(m-1)}$ is $-2\alpha^2$, which could very well be equal to $-1$, so we need to take this into account.
 Clearly
\begin{equation}\nonumber
\int_0^{s_{\sigma(m-2)}} (s_{\sigma(m-1)}+\epsilon)^xds_{\sigma(m-1)}=\mathcal O\left( \log(s_{\sigma(m-2)}/\epsilon+3)(\epsilon^{x+1}+(s_{\sigma(m-2)}+\epsilon)^{x+1})\right)\end{equation}
for any fixed $x$ as $\epsilon\to 0$. Since 
\begin{equation}\nonumber
\log(s_{\sigma(m-2)}/\epsilon+3)\leq \log(s_{\sigma(1)}/\epsilon+3),\end{equation}
it follows that

\begin{multline}\label{intm2}
\int_{\mathcal W_\sigma} \, \prod_{\ell=1}^{m-1}  \left( s_{\sigma(\ell)}+\epsilon\right)^{-2\alpha^2(m-\ell)} ds_{\sigma(m-1)}\dots ds_{\sigma(1)}=
\mathcal O\Bigg(
\int \log(s_{\sigma(1)}/\epsilon+3) \\ \times \prod_{\ell=1}^{m-2}  \left( s_{\sigma(\ell)}+\epsilon\right)^{-2\alpha^2(m-\ell)} \left(
\epsilon^{-2\alpha^2+1}+\left(s_{\sigma(m-2)}+\epsilon\right)^{-2\alpha^2+1}
\right)ds_{\sigma(m-2)}\dots ds_{\sigma(1)}
\Bigg),
\end{multline}
as $\epsilon\to 0$
where integration on the right hand side is taken over $0<s_{\sigma(m-2)}<\dots<s_{\sigma(1)}<1$.
We will next integrate out $s_{\sigma(m-2)}$, then $s_{\sigma(m-3)}$, etc. To do this, we introduce the following notation  for $v=1,2,\dots,m-2$:
\begin{multline}\nonumber
J_\epsilon(v)=\int \left[\log(s_{\sigma(1)}/\epsilon+3)\right]^v \prod_{\ell=1}^{m-v-1}  
\left( s_{\sigma(\ell)}+\epsilon\right)^{-2\alpha^2(m-\ell)} \\ \times
\left[
\sum_{r=0}^v\epsilon^{r-2\alpha^2\sum_{j=1}^rj}\left(s_{\sigma(m-v-1)}+\epsilon\right)^{v-r-2\alpha^2\sum_{j=r+1}^vj}
\right]
ds_{\sigma(m-v-1)}\dots ds_{\sigma(1)},
\end{multline}
with integration taken over $0<s_{\sigma(m-v-1)}<\dots <s_{\sigma(1)}<1$, and
where we interpret $\sum_{j=1}^0j=\sum_{j=v+1}^vj=0$. We observe that the error term on the right hand side of \eqref{intm2} is equal to $J_\epsilon(1)$. It is easily verified that
\begin{equation}\nonumber
J_\epsilon(v)=\mathcal O\left(J_\epsilon(v+1)\right),
\end{equation}
as $\epsilon\to 0$, for $v=1,2,\dots,m-3$.
  Iterating, we obtain
\begin{multline}\nonumber
\int_{\mathcal W_\sigma} \, \prod_{\ell=1}^{m-1}  \left( s_{\sigma(\ell)}+\epsilon\right)^{-2\alpha^2(m-\ell)} ds_{\sigma(m-1)}\dots ds_{\sigma(1)}=\mathcal O(J_\epsilon(m-2))=\\
\mathcal O\Bigg(
\int_0^1 \left[\log(s/\epsilon+3)\right]^{m-2} 
\left( s+\epsilon\right)^{-2\alpha^2(m-1)} 
\left[\sum_{r=0}^{m-2}\epsilon^{r-2\alpha^2\sum_{j=1}^rj}\left( s+\epsilon\right)^{m-2-r-2\alpha^2\sum_{r+1}^{m-2}j}
\right]
 ds
\Bigg),
\end{multline}
as $\epsilon\to 0$, where we interpret $\sum_{j=1}^0j=\sum_{j=m-1}^{m-2}j=0$. 
 Since $m\alpha^2>1$, it follows that the power of $s+\epsilon$ is smaller than $-1$, namely:
\begin{equation}\nonumber
-2\alpha^2(m-1)+m-2-r-2\alpha^2\sum_{r+1}^{m-2}j<-1,
\end{equation}
for any $r=0,1,2,\dots,m-2$. If $x<-1$, then
\begin{equation}\nonumber
\int_0^1(s+\epsilon)^{x}\log(s/\epsilon+3)^{m-2}ds=\mathcal O(\epsilon^{x+1}),
\end{equation}
as $\epsilon\to 0$, and thus it follows that
\begin{equation}\nonumber
\int_{\mathcal W_\sigma} \, \prod_{\ell=1}^{m-1}  \left( s_{\sigma(\ell)}+\epsilon\right)^{-2\alpha^2(m-\ell)} ds_{\sigma(m-1)}\dots ds_{\sigma(1)}=\mathcal O\left(\epsilon^{(m-1)(1-m\alpha^2)}\right),
\end{equation}
as $\epsilon\to 0$.
Thus by \eqref{UBIeps} and \eqref{ints1m1}, $I_\epsilon(\alpha)=\mathcal O\left(\epsilon^{(m-1)(1-m\alpha^2)}\right),$ as $\epsilon\to 0$, which combined with the lower bound \eqref{LBIeps} and Theorem \ref{Theorem1} proves Corollary \ref{Corollary}.

\subsubsection{$m\alpha^2=1$}
We start by finding a lower bound for $I_\epsilon(\alpha)$ when $m\alpha^2=1$. Let 
\begin{equation}\nonumber
D_\epsilon(j)=\{0<t_1<1: \,\, t_1<t_k<t_1+2j\epsilon \textrm{ for all } k=2,\dots,m\},\end{equation}
and let $B_\epsilon(j)=D_\epsilon(j)\setminus D_\epsilon(j-1)$.
On $B_\epsilon(j)$, the integrand in \eqref{defIalpha} satisfies
\begin{equation}\nonumber
((j+1)\epsilon)^{-m(m-1)\alpha^2}\leq 
\prod_{1\leq j<k\leq m} \left( \sin \left|\frac{t_j-t_k}{2}\right|+\epsilon\right)^{-2\alpha^2},
\end{equation}
for $j=1,2,\dots,1/\epsilon$, assuming for  ease of notation that $1/\epsilon$ is an integer.
Combined with the fact that $B_\epsilon(j)$ are disjoint for $j=1,2,\dots$ and the fact that $m\alpha^2=1$, it follows that
\begin{equation}\nonumber
\sum_{j=1}^{1/\epsilon} ((j+1)\epsilon)^{-(m-1)}\int_{B_\epsilon(j)}dt_1\dots dt_m\leq I_\epsilon(\alpha).
\end{equation}
Since
\begin{equation}\nonumber
\int_{B_\epsilon(j)}dt_1\dots dt_m=(2\epsilon)^{m-1}\left(j^{m-1}-(j-1)^{m-1}\right)\geq \epsilon^{m-1}(j+1)^{m-2},\end{equation}
for $j$ sufficiently large, say $j>j_0$, it follows that
\begin{equation}\label{LBIeps2}
\log (1/\epsilon)-\hat c_0<\sum_{j_0=1}^{1/\epsilon} (j+1)^{-1}\leq I_\epsilon(\alpha),
\end{equation}
for some constant $\hat c_0$. Thus we have a lower bound for $I_\epsilon(\alpha)$ and we look to obtain a corresponding upper bound.

We observe that the upper bounds \eqref{UBIeps} and \eqref{ints1m1} hold also for $m\alpha^2=1$. If $m=2$ and $m\alpha^2=1$, then the right hand side of \eqref{ints1m1} is of order $\log \epsilon^{-1}$ as $\epsilon \to 0$. We assume that $m>2$ and integrate in the variable $s_{\sigma(m-1)}$. We have $2\alpha^2=2/m<1$, and it follows that
\begin{equation}\nonumber
\int_0^{s_{\sigma(m-2)}}\left(s_{\sigma(m-1)}+\epsilon\right)^{-2\alpha^2}ds_{\sigma(m-1)}=\mathcal O\left(\left(
s_{\sigma(m-2)}+\epsilon \right)^{-2\alpha^2+1}\right),
\end{equation}
as $\epsilon \to 0$, and thus
\begin{multline}\nonumber
\int_{\mathcal W_\sigma} \, \prod_{\ell=1}^{m-1}  \left( s_{\sigma(\ell)}+\epsilon\right)^{-2\alpha^2(m-\ell)} ds_{\sigma(m-1)}\dots ds_{\sigma(1)}=\\
\mathcal O\left(
\int  \prod_{\ell=1}^{m-2}  \left( s_{\sigma(\ell)}+\epsilon\right)^{-2\alpha^2(m-\ell)} \left(
s_{\sigma(m-2)}+\epsilon \right)^{-2\alpha^2+1}
ds_{\sigma(m-2)}\dots ds_{\sigma(1)}
\right),
\end{multline}
as $\epsilon\to 0$
where integration on the right hand side is taken over $0<s_{\sigma(m-2)}<\dots<s_{\sigma(1)}<1$.
We redefine $J_\epsilon(v)$ as follows
\begin{equation}\nonumber
J_\epsilon(v)=\int  \prod_{\ell=1}^{m-v-1}  
\left( s_{\sigma(\ell)}+\epsilon\right)^{-2\alpha^2(m-\ell)} 
\left(s_{\sigma(m-v-1)}+\epsilon\right)^{v-\alpha^2v(v+1)}
ds_{\sigma(m-v-1)}\dots ds_{\sigma(1)},
\end{equation}
with integration taken over $0<s_{\sigma(m-v-1)}<\dots <s_{\sigma(1)}<1$. Since $m\alpha^2=1$, we note that the power of $s_{\sigma(m-v-1)}$ is greater than $-1$ for  $v=1,2,\dots,m-3$, namely:
\begin{equation}\nonumber
-2\alpha^2(v+1)+v-\alpha^2v(v+1)>-1,
\end{equation}
and it follows that
$J_\epsilon(v)=\mathcal O\left(J_\epsilon(v+1)\right)$ as $\epsilon \to 0$ for such $v$. It follows that
\begin{multline}\nonumber
\int_{\mathcal W_\sigma} \, \prod_{\ell=1}^{m-1}  \left( s_{\sigma(\ell)}+\epsilon\right)^{-2\alpha^2(m-\ell)} ds_{\sigma(m-1)}\dots ds_{\sigma(1)}=\mathcal O\left(J_\epsilon(m-2)\right)=
\\ \mathcal O\left(\int_0^1\left(s+\epsilon\right)^{m-2-\alpha^2m(m-1)}ds\right),
\end{multline}
as $\epsilon\to 0$. Since $m\alpha^2=1$, it follows that the right hand side is simply $\mathcal O\left(\log \epsilon^{-1}\right)$. Thus, combined with \eqref{LBIeps2} and Theorem \ref{Theorem1}, we have proven Corollary \ref{Corollary} for $m\alpha^2=1$.

\section{Statistics of impenetrable bosons in 1 dimension}\label{SecBos}

Consider
\begin{equation}\nonumber
\psi(x_1,\dots,x_n)=\frac{1}{\sqrt{n!L^n}}\prod_{1\leq j<k\leq n}\left|e^{2\pi i x_j/L}-e^{2\pi i x_k/L}\right|.
\end{equation}
It was proven by Girardeau \cite{Gir} that it has the following properties:
\begin{itemize}\item $\psi$ is the ground--state solution to the general time-independent Schr\"odinger equation in one-dimension with $n$ particles.
\item $\psi$ is symmetric with respect to interchange of $x_i$ and $x_j$ for $i\neq j$ (Bose--Einstein statistics).
\item $\psi$ is translationally invariant with period $L$.
\item $\psi$ vanishes when $x_i=x_j$ for $i\neq j$ (mutual impenetrabililty of particles).
\end{itemize}

In fact Girardeau only proved the above for odd $n$, but as noted by Lieb and Liniger \cite{Lieb1} (footnote 6), it is equally valid for even $n$.
When the system is in ground state, the wave function $\psi$ gives rise to a probability distribution for both the position and momentum. The position of the particles on $[0,L)$ has joint probability density function $\psi(x_1,\dots,x_n)^2$. Following the footsteps of Girardeau, we take as our starting point that the  wave function for the momentum is given by the Fourier transform of the wave function of the position:
\begin{equation} \nonumber
\phi(\mathcal M_1,\dots,\mathcal M_n)=\frac{1}{\sqrt{L^n}}\int_{(0,L)^n}dx_1\dots dx_n \psi(x_1,\dots,x_n)e^{-2\pi i \sum_{j=1}^nx_j\mathcal M_j/L},\qquad \mathcal M_j\in \mathbb Z.
\end{equation}
Thus the probability of the $j$'th particle having momentum $2\pi \mathcal M_j/L$ for each $j=1,\dots, n$ is given by 
\begin{equation}\label{momentumpdf}|\phi(\mathcal M_1,\dots,\mathcal M_n)|^2.\end{equation} 
It is easily verified that $\phi(\mathcal M_1,\dots,\mathcal M_n)$ is independent of $L$,  and that 
\begin{equation}
\sum_{\mathcal M_1,\dots,\mathcal M_n\in \mathbb Z}|\phi(\mathcal M_1,\dots,\mathcal M_n)|^2 =1.\end{equation}
Thus $|\phi|^2$ may simply be viewed as a probability distribution on $\mathbb Z^n$, which is the viewpoint we will take in Corollary \ref{CorMom} below, where we fix $L=2\pi$ without loss of generality.

 Since the particles are indistinguishable from one another, it is preferable to characterize the distribution as a point process, which we do as follows. Let $N_{\mathcal M}(n)$ denote the number of particles with momentum $2\pi \mathcal M/L$. Then if  $\mathcal M_1,\dots, \mathcal M_k$ are distinct, it follows from \eqref{momentumpdf} by a straightforward calculation that
\begin{equation} \nonumber
\mathbb E\left[ \prod_{j=1}^k N_{\mathcal M_j}(n)\right]=\pi_{k,n}(\mathcal M_1,\dots, \mathcal M_k),
\end{equation}
where
\begin{equation}\label{defpikn}
\pi_{k,n}(\mathcal M_1,\dots, \mathcal M_k)=\frac{n!}{(n-k)!}\sum_{\mathcal M_{k+1},\dots, \mathcal M_{n}\in \mathbb Z} |\phi(\mathcal M_1,\dots, \mathcal M_{n})|^2.
\end{equation} 
The above is only valid for distinct particles, for moments of $N_\mathcal M$ we have
\begin{equation}\label{varNM}
\pi_{k,n}(\mathcal M,\dots,\mathcal M)=\mathbb E\left[N_{\mathcal M}(N_{\mathcal M}-1)\dots(N_{\mathcal M}-k+1)\right],
\end{equation}
where $N_{\mathcal M}=N_{\mathcal M}(n)$.
Then the expected number of particles with $0$ momentum is given by $\pi_{1,n}(0)$. In 1963, Schultz \cite{Sch} proved that $\pi_{1,n}(0)=\mathcal O (n^{-\pi/4})$ as $n \to \infty$, which shows that there is no Bose-Einstein condensation according to the Penrose-Onsager criterion (the criterion states that  if the proportion of the particles  expected to have $0$ momentum tends to 0 as $n\to \infty$, then there is no Bose-Einstein condensation).  The upper bound obtained by Schultz was not optimal.  In 1964 Lenard \cite{Lenard} was able to improve on this, and obtained that $\mathbb E(N_0(n))=\mathcal O(n^{1/2})$ as $n\to \infty$. Lenard's approach was to make
a connection to Toeplitz determinants with Fisher--Hartwig singularities  by observing that if we denote the $k$ particle reduced density matrix by
\begin{equation}\label{defrhokn}
 \rho_{k,n}(x_1,\dots, x_k, y_1,\dots, y_k)=\int dx_{k+1}\dots dx_n \psi(x_1,\dots,x_n)\psi(y_1,\dots,y_k,x_{k+1},\dots,x_n),
\end{equation}
then $ \rho_{k,n}$ is a Toeplitz determinant with $2k$ FH singularities. This observation relies on the   multiple integral formula \eqref{intrep}.
By \eqref{defpikn} and \eqref{defrhokn} it is easily verified that 
\begin{multline}\label{defpik}
\pi_{k,n}(\mathcal M_1,\dots, \mathcal M_k)=\frac{n!}{(n-k)!L^k}\int_{(0,L)^{2k}}dx_1\dots dx_k \,dy_1 \dots dy_k \\  e^{2\pi i\sum_{j=1}^k (x_j-y_j)\mathcal M_j/L}\rho_{k,n}(x_1,\dots, x_k, y_1,\dots, y_k).
\end{multline}
Thus, to obtain the asymptotics of $\pi_{k,n}$ one must obtain those of $\rho_{k,n}$. The asymptotics of \\ $ \rho_{k,n}(x_1,\dots, x_k, y_1,\dots, y_k)$ was studied in the limit $n\to \infty$ with $L=n$ and for fixed $x_j, y_j$ independent of $n$ in \cite{VadTr,JMMS}. This is equivalent to studying Toeplitz determinants with $2k$ FH singularities with $\alpha_j=1/2$ for $j=1,\dots, 2k$, in the double scaling limit where the singularities are all at a distance of length $\mathcal O(1/n)$ from each other. This gave rise to some of the first connections to Painlev\'{e} V in the study of Toeplitz determinants. To obtain more detailed asymptotics for $\pi_{k,n}$ however, uniform asymptotics of $\rho_{k,n}$ are required. As mentioned in the introduction, Claeys and Krasovsky \cite{CK} obtained uniform asymptotics for $\rho_{1,n}$, and they  relied on \eqref{formulaCK} to prove that
\begin{equation}\nonumber
\mathbb E\left[\frac{N_0(n)}{\sqrt{n}}\right]\to \frac{\sqrt{2}}{\pi}G(3/2)^4\int_0^{\pi/2}(\sin t)^{-1/2} dx,
\end{equation}
as $n\to \infty$ (see formula (1.53) of \cite{CK}).

We are interested in not just the expectation of $N_0(n)/\sqrt{n}$, but also the variance and higher moments. By combining \eqref{formulaWidom}--\eqref{FormulaE} with Theorem \ref{Theorem1}, we obtain the following.

\begin{corollary} \label{CorMom} Fix $L=2\pi$ and let $\mathcal M_1,\dots, \mathcal M_n \in \mathbb Z$ be random variables with the probability distribution
\begin{equation}\nonumber
\textrm{Prob}(\mathcal M_1,\dots, \mathcal M_n)=|\phi(\mathcal M_1,\dots, \mathcal M_n)|^2.\end{equation}
Then, as $n\to\infty$,
\begin{equation}\nonumber
\mathbb E\left[\left(\frac{N_0(n)}{\sqrt n}\right)^k\right]\to\frac{G(3/2)^{4k}}{(2\pi)^{2k}}  \int_{(0,2\pi)^{2k}}
\frac{\prod_{1\leq r<s\leq k}\left|e^{it_r}-e^{it_s}\right|\prod_{k+1\leq r<s\leq 2k}\left|e^{it_r}-e^{it_s}\right|}{\sqrt{\prod_{1\leq r<s\leq 2k}\left|e^{it_r}-e^{it_s}\right|}}dt_1\dots dt_{2k}, 
\end{equation}
where
\begin{equation}\nonumber N_0(n)=\#\{j:\mathcal M_j=0,\}_{j=1}^n.\end{equation}
\end{corollary}
\begin{proof}
To study higher moments of $N_0=N_0(n)$, we have by \eqref{varNM} that
\begin{equation}\label{piE}
\mathbb E\left[N_0^k\right]=\pi_{k,n}(0,\dots,0)+\mathcal O\left(\max_{j=1,\dots,k-1}\pi_{j,n}(0,\dots,0)\right),
\end{equation}
for any fixed $k$.

We note that $\pi_{k,n}(0,\dots,0)$ is independent of $L$, and so we set $L=2\pi$. Thus, by \eqref{defrhokn}-\eqref{defpik},
\begin{equation}\nonumber
\pi_{k,n}(0,\dots,0)=\frac{1}{(2\pi)^{2k}}\int_{(0,2\pi)^{2k}}D_{n-k}\left(f_{2k}^{(1/2)}\right)\prod_{1\leq r<s\leq k}\left|e^{it_r}-e^{it_s}\right|\prod_{k+1\leq r<s\leq 2k}\left|e^{it_r}-e^{it_s}\right|dt_1\dots dt_{2k},
\end{equation}
where we recall the notation \eqref{falpha}.
Thus $\pi_{k,n}(0,\dots,0)$ is evaluated by combining \eqref{formulaWidom}--\eqref{FormulaE} and Theorem \ref{Theorem1}   as $n\to \infty$ by using similar types of arguments as in Section \ref{malp<1}, and we obtain
\begin{multline} \nonumber
\pi_{k,n}(0,\dots,0)=n^{k/2}\frac{G(3/2)^{4k}}{(2\pi)^{2k}} \\ \times \int_{(0,2\pi)^{2k}}
\frac{\prod_{1\leq r<s\leq k}\left|e^{it_r}-e^{it_s}\right|\prod_{k+1\leq r<s\leq 2k}\left|e^{it_r}-e^{it_s}\right|}{\sqrt{\prod_{1\leq r<s\leq 2k}\left|e^{it_r}-e^{it_s}\right|}}dt_1\dots dt_{2k}(1+o(1)).
\end{multline}
Thus the corollary follows from \eqref{piE}.
\end{proof}

\section{Method of proof of Theorem \ref{Theorem1}} \label{SecMethod}
Denote $\psi_0(z)=\chi_0=1/\sqrt{D_1(f)}$ and define the polynomials $\psi_j$ for $j=1,2,\dots$ by
\begin{equation}\nonumber
\psi_j(z)=\frac{1}{\sqrt{D_j(f)D_{j+1}(f)}}\det 
{ \begin{pmatrix} f_0 & f_{-1} & \dots &f_{-j+1} &f_{-j}\\
f_1& f_0 &\dots&f_{-j+2}&  f_{-j+1}\\
&&\ddots &\\
f_{j-1}&f_{j-2}& \dots& f_0 &f_{-1}\\
1&z&\dots &z^{j-1}&z^j \end{pmatrix}}=\chi_jz^j+\dots,
\end{equation}
where the leading coefficient $\chi_j$ is given by
\begin{equation}
\chi_j=\sqrt{\frac{D_j(f)}{D_{j+1}(f)}}.\label{DefChij}
\end{equation}
By the representation \eqref{intrep}, it follows that $D_j(f)>0$ and we fix $\chi_j>0$.
It is easily seen that $\psi_j$ are orthonormal on the unit circle:
\begin{equation}\nonumber
\int_0^{2\pi} \psi_k\left(e^{i\theta}\right)\overline{ \psi_j\left(e^{i\theta}\right)} f\left(e^{i\theta}\right)\frac{d\theta}{2\pi}= \delta_{jk}=\begin{cases}
0&{\rm for}\, j\neq k,\\
1&{\rm for } \, j=k,
\end{cases}
\end{equation}
for $j,k=0,1,2,\dots$.
By \eqref{DefChij} and the definition $\chi_0=1/\sqrt{D_1(f)}$,
\begin{equation} \label{ToepOP}
D_n(f)=\prod_{j=0}^{n-1}\chi_{j}^{-2}.
\end{equation}

Given $U>\epsilon>0$, we say that the parameters $t_1,\dots, t_m$ satisfy condition $(\epsilon,U,n)$ if $t_1,\dots,t_m\in \mathcal S_t$, where 
\begin{equation}\mathcal S_t=\{t_1,\dots,t_m:0\leq t_1<\dots<t_m<2\pi-\pi/m\},\end{equation} and  for each $1\leq j<k\leq m$, either $t_k-t_j< \epsilon/n$ or $t_k-t_j\geq U/n$. The  assumption that $t_m<2\pi-\pi/m$ one can make without loss of generality when studying Toeplitz determinants, since the Toeplitz determinant is rotationally invariant (i.e. $D_n(f(e^{i\theta}))=D_n(f(e^{i(\theta+x)}))$ for all $x\in [0,2\pi)$).

If $t_1,\dots,t_m$ satisfies condition $(\epsilon,U,n)$, the points $\{t_1,\dots,t_m\}$ partition naturally into clusters 
$\textbf{Cl}_1(\epsilon,U,n),\dots, \textbf{Cl}_r(\epsilon,U,n)$, where $r=r(\epsilon,U,n)$, satisfying the following conditions.
\begin{itemize}
\item  The radius of each cluster is less than $\epsilon/n$. Namely, $\epsilon_n<\epsilon$, where 
\begin{equation}\nonumber
\epsilon_n=n\max_{j=1,\dots,r}\, \max_{x,y\in {\bf Cl}_j(\epsilon,U,n)} |x-y|.
\end{equation}
\item The distance between any two clusters is greater than $U/n$. Namely, $\widehat u_n>U$, where
\begin{equation}\label{hatun}
\widehat u_n=n\min_{1\leq j<k\leq r} \, \min_{\substack {x\in {\bf Cl}_j\\ y\in {\bf Cl}_k}} |x-y|.
\end{equation}
 \end{itemize}

In Sections \ref{ModelProblem}-- \ref{SectionAsym}, we prove the following proposition.
\begin{Prop}\label{PropPoly}
\begin{itemize}
\item[(a)]
As $n\to \infty$,
\begin{equation}\nonumber
\log \chi_n=-V_0/2 +\mathcal O(1/n),
\end{equation}
uniformly for $t_1,\dots,t_m\in \mathcal S_t$.
\item[(b)]
There exists $U_1>U_0>0$, $C>0$ and $n_0>0$, such that if
the parameters $t_1,\dots,t_m$ satisfy condition $(U_0,U_1,n)$ and $n>n_0$, then
\begin{equation}\nonumber
\left|\log \chi_n+V_0/2 +H_n(\alpha_j,\beta_j,t_j)_{j=1}^m\right|<C\left(\frac{1}{n^2}+\frac{1}{ne^{\widehat u_n}}+\frac{\epsilon_n}{n}\right),
\end{equation}
where
\begin{equation}\nonumber
H_n(\alpha_j,\beta_j,t_j)_{j=1}^m=
\frac{1}{2n}\sum_{j=1}^m \left( \alpha_j^2 - \beta_j^2\right)+\frac{1}{n}\sum_{1\leq j<k\leq m}(\alpha_j\alpha_k-\beta_j\beta_k)\mathds 1_{U_0/n}(|t_k-t_j|),\end{equation}
 where
\begin{equation}\nonumber
\mathds 1_{U_0/n}(x)=\begin{cases}
1&0<x<U_0/n,\\
0&U_0/n<x.
\end{cases}
\end{equation}
\end{itemize}
\end{Prop}

Using Proposition \ref{PropPoly}, we now compute the asymptotics of $D_n(f)$ as $n\to \infty$, for a specific configuration $0\leq t_1<t_2<\dots<t_m<2\pi-\pi/m$, but with error terms which are uniform over all configurations. Let $n_0$ be a fixed positive integer such that the asymptotics of Proposition \ref{PropPoly} are valid for $n\geq n_0$. Then $D_{n_0}(f)$ is a continuous function in terms of $t_j$ on the compact set $t_1,\dots,t_m\in [0,2\pi]$, and is thus uniformly bounded as $t_j$ vary.  Thus by \eqref{ToepOP}
\begin{equation}\label{bound0}
\log D_n=-2\sum_{N=n_0}^{n-1}\log \chi_N+\mathcal O(1),\end{equation}
as $n\to \infty$, uniformly over $\mathcal S_t$. Denote $\mathbb N_0=\mathbb N \setminus \{0,1,2,\dots,n_0\}$,
 and let 
\begin{equation}\nonumber
J_t=\{N\in \mathbb N_0:t_1,\dots, t_m \textrm{ satisfy condition } (U_0,U_1,N)\},
\end{equation}
with complement $J^c_t=\mathbb N_0\setminus J_t$. Then 
\begin{equation}\nonumber
J^c_t=\cup_{j=1}^{m-1}I_{j,t},\qquad I_{j,t}= \{N\in \mathbb N_0: (t_{j+1}-t_j)N\in [U_0,U_1)\}. \end{equation}
Written differently, we have
\begin{equation}\nonumber
I_{j,t}=\left\{ N\in \mathbb N_0: \frac{U_0}{t_{j+1}-t_j}\leq N<\frac{U_1}{t_{j+1}-t_j}\right\},
\end{equation}
and it follows that
\begin{equation}\nonumber
\sum_{N\in I_{j,t}}\frac{1}{N}=\log (U_1/U_0)+\mathcal O(1),\end{equation}
uniformly for $t_{j+1}-t_j>0$. Since $U_0$ and $U_1$ are fixed, it follows that the right hand side is uniformly bounded, and by Proposition \ref{PropPoly} (a) and the fact that $H_N=\mathcal O(1/N)$ we have
\begin{equation}\label{bound1}
\sum_{N\in J^c_t}\left(\log \chi_n+V_0/2+H_N(\alpha_j,\beta_j,t_j)_{j=1}^m\right)=\mathcal O(1),
\end{equation}
uniformly for $t_1,\dots,t_m\in[0,2\pi-\pi/m)$.

Suppose that $t_1,\dots,t_m$ satisfy condition $(U_0,U_1,N)$ for $N$ in an interval $N_1,N_1+1,\dots,N_2$. By Proposition \ref{PropPoly} (b), 
\begin{equation}\label{sumcont}
\sum_{N_1}^{N_2}\left(\log \chi_n+V_0/2+H_N(\alpha_j,\beta_j,t_j)_{j=1}^m\right)<C\sum_{N=N_1}^{N_2}\left(\frac{1}{N^2}+\frac{1}{Ne^{\widehat u_N}}+\frac{\epsilon_N}{N}\right),
\end{equation}
Since $\widehat u_N=N(t_i-t_{i-1})$ for some $i\in\{2,3,\dots,m\}$ (where $i$ is fixed for $N\in [N_1,N_2]$), it follows that $\frac{\widehat u_N}{N}=\frac{\widehat u_{N_1}}{N_1}$, and as a consequence (bearing in mind that $\frac{\widehat u_{N_1}}{N_1}>U_1$) we have $\widehat u_N>NU_1$. Thus, bounding the sum by a suitable integral,
\begin{equation}
\sum_{N=N_1}^{N_{2}}\frac{1}{Ne^{\widehat u_N}}\leq \frac{e^{-U_1}}{U_1}.
\end{equation}
 Similarly, $\epsilon_N/N=\epsilon_{N_2}/N_2$, and thus
\begin{equation}\label{sumcont1}
\sum_{N=N_1}^{N_{2}}\frac{\epsilon_N}{N}\leq \epsilon_{N_2}\leq U_0.
\end{equation}
Since $J_t^c$ is composed of at most $m-1$ disjoint intervals, it follows that $J_t$ is composed of at most $m$ disjoint intervals, and it follows by \eqref{sumcont}-\eqref{sumcont1} that
\begin{equation}\label{bound2}
\sum_{\substack{N\in J_t\\ N\leq n}}\left(\log \chi_n+V_0/2+H_N(\alpha_j,\beta_j,t_j)_{j=1}^m\right)\leq Cm\left(\frac{1}{n_0}+\frac{e^{-U_1}}{U_1}+U_0\right).
\end{equation}
Since $U_0, U_1,n_0$ are just arbitrary constants, the right hand side is bounded uniformly over $\mathcal S_t$. Thus, by \eqref{bound0}, \eqref{bound1}, \eqref{bound2}, it follows that
\begin{equation}\nonumber
\log D_n(f)=nV_0+2\sum_{N=1}^nH_N(\alpha_j,\beta_j,t_j)_{j=1}^m+\mathcal O(1),\end{equation}
uniformly over $\mathcal S_t$.
 Since
\begin{equation}\nonumber
\sum_{N=1}^n\frac{1}{N} \mathds 1_{U_0/N}(t_j-t_i)=\log\frac{1}{t_j-t_i+1/n}+\mathcal O(1),
\end{equation}
as $n\to \infty$, with the implicit constant depending only on $U_0$ which is fixed, it follows that
\begin{equation}\nonumber
\log D_n(f)=nV_0+\sum_{j=1}^m\left(\alpha_j^2-\beta_j^2\right)\log n+\sum_{1\leq j< k\leq m}2(\alpha_j\alpha_k-\beta_j\beta_k)\log \left(\frac{1}{\left|t_j-t_k\right|+n^{-1}}\right)+\mathcal O(1),
\end{equation}
as $n\to \infty$, uniformly over $\mathcal S_t$. For such $t_1,\dots,t_m$, we have
\begin{equation}
\log \left(\frac{1}{\left|t_j-t_k\right|+n^{-1}}\right)=\log \left(\frac{1}{\sin \frac{\left|t_j-t_k\right|}{2}+n^{-1}}\right)+\mathcal O(1)\end{equation}
with uniform error terms, which yields Theorem \ref{Theorem1} for
$t_1,\dots,t_m\in \mathcal S_t$, and the full theorem 
follows from the aforementioned rotational invariance of the Toeplitz determinant.

\subsection{Structure of the proof of Proposition \ref{PropPoly}}

We will prove the following proposition, which holds if and only Proposition \ref{PropPoly} (a) holds.
\begin{Prop}\label{PropPoly2} 
Given $u>0$, there exists $\widetilde U>u$, $\widetilde C>0$, and $\widetilde N_0>0$ such that if 
the parameters $t_1,\dots,t_m$ satisfy condition $(u,\widetilde U,n)$ and $n>\widetilde N_0$, then
\begin{equation}\nonumber
\left|\log \chi_n+V_0/2 \right|<\widetilde C/n.
\end{equation}
\end{Prop}

We now show that Proposition \ref{PropPoly2} implies Proposition \ref{PropPoly} (a). It will be useful make the dependence of the implicit constants in Proposition \ref{PropPoly2} more explicit, so we denote $\widetilde U=\widetilde U(u)$, $\widetilde C=\widetilde C(u)$, and $\widetilde N_0=\widetilde N_0(u)$.
Let $U_0,U_1$ be two constants for which Proposition \ref{PropPoly} (b) holds.
For $j=1,2,\dots,m-1$, we define $U_{j+1}=\widetilde U (U_j)$ (meaning that given $u=U_j$,  Proposition \ref{PropPoly2} holds with $\widetilde U=\widetilde U(U_j)$, which we define to be equal to $U_{j+1}$).
For each $N=1,2,\dots,n$, we have a sequence of conditions
\begin{equation}\label{sequence}
(U_0,U_1,N),\quad (U_1,U_2,N), \quad (U_2,U_3,N), \dots, (U_{m-1},U_{m},N).
\end{equation} 
Then the configuration $t_1,\dots,t_m$ will satisfy one of the conditions in the sequence \eqref{sequence}, because otherwise, for each $k=1,2,\dots,m$, one would have  a corresponding $j=j(k)=2,3,\dots,m$ such that $U_{k-1}/N\leq  t_j-t_{j-1}<  U_k/N$, meaning that there are at least $m+1$ distinct points in $\{t_j\}_{j=1}^m$, which is a contradiction. Thus, if we denote the maximum of the implicit constants $\widetilde C(U_j)$ and $\widetilde N_0(U_j)$ over $j=1,\dots, m-1$ by $\widetilde C_{\textrm{max}}$ and $\widetilde N_{0, \textrm{max}}$ we obtain
\begin{equation}\nonumber
\left|\log \chi_n+V_0/2 \right|<\widetilde C_{\textrm{max}}/n,
\end{equation}
for any $n\geq \widetilde N_{0, \textrm{max}}$, which proves Proposition \ref{PropPoly} (a).

We will prove Proposition \ref{PropPoly} (b) and Proposition \ref{PropPoly2} by applying the Deift--Zhou steepest descent analysis \cite{DeiftZhou} to a Riemann--Hilbert (RH) problem associated to the orthogonal polynomials $\psi_j$. Under the Deift--Zhou steepest descent framework, there are several standard ingredients, including the opening of the lens, and the construction of a main parametrix and local parametrices. Among these ingredients, the opening of the lense and  the construction of a local parametrix is the most involved. Each local parametrix  contains a cluster ${\bf Cl}_j(u,\widetilde U,n)$, and we map a model RH problem to a shrinking disc containing ${\bf Cl}_j(u,\widetilde U,n)$. We construct and analyze the model RH problem in the next section, Section \ref{ModelProblem}, and use these results in Section \ref{SectionAsym} to prove Propositions \ref{PropPoly} (b) and \ref{PropPoly2}.

\section{Model RH problem}\label{ModelProblem}

\begin{figure}
\begin{tikzpicture}
\draw [decoration={markings, mark=at position 0.12 with {\arrow[thick]{>}}},
        postaction={decorate}]
        [decoration={markings, mark=at position 0.38 with {\arrow[thick]{>}}},
        postaction={decorate}]
        [decoration={markings, mark=at position 0.56 with {\arrow[thick]{>}}},
        postaction={decorate}]
        [decoration={markings, mark=at position 0.69 with {\arrow[thick]{>}}},
        postaction={decorate}] [decoration={markings, mark=at position 0.87 with {\arrow[thick]{>}}},
        postaction={decorate}] (0,-2)--(0,6);
\draw [decoration={markings, mark=at position 0.5 with {\arrow[thick]{>}}},
        postaction={decorate}] (-2,-4.5)--(0,-2);
\draw [decoration={markings, mark=at position 0.5 with {\arrow[thick]{>}}},
        postaction={decorate}] (2,-4.5)--(0,-2);
\draw [decoration={markings, mark=at position 0.5 with {\arrow[thick]{>}}},
        postaction={decorate}] (0,6)-- (-2,8.5);
\draw [decoration={markings, mark=at position 0.5 with {\arrow[thick]{>}}},
        postaction={decorate}] (0,6)-- (2,8.5);
\draw [decoration={markings, mark=at position 0.20 with {\arrow[thick]{>}}},
        postaction={decorate}] [decoration={markings, mark=at position 0.80 with {\arrow[thick]{>}}},
        postaction={decorate}] (-2.5,4)--(2.5,4);
\draw [decoration={markings, mark=at position 0.20 with {\arrow[thick]{>}}},
        postaction={decorate}] [decoration={markings, mark=at position 0.80 with {\arrow[thick]{>}}},
        postaction={decorate}] (-2.5,3)--(2.5,3);
\draw [decoration={markings, mark=at position 0.20 with {\arrow[thick]{>}}},
        postaction={decorate}] [decoration={markings, mark=at position 0.80 with {\arrow[thick]{>}}},
        postaction={decorate}] (-2.5,2)--(2.5,2);
\draw [decoration={markings, mark=at position 0.20 with {\arrow[thick]{>}}},
        postaction={decorate}] [decoration={markings, mark=at position 0.80 with {\arrow[thick]{>}}},
        postaction={decorate}] (-2.5,0)--(2.5,0);
\node [above] at (2.5,0) {$e^{\pi i (\alpha_4+\beta_4)\sigma_3}$}; 
\node [above] at (2.5,2) {$e^{\pi i (\alpha_3+\beta_3)\sigma_3}$};       
\node [above] at (2.5,3) {$e^{\pi i (\alpha_2+\beta_2)\sigma_3}$};
\node [above] at (2.5,4) {$e^{\pi i (\alpha_1+\beta_1)\sigma_3}$};
\node [above] at (-2.5,0) {$e^{\pi i (\alpha_4-\beta_4)\sigma_3}$}; 
\node [above] at (-2.5,2) {$e^{\pi i (\alpha_3-\beta_3)\sigma_3}$};       
\node [above] at (-2.5,3) {$e^{\pi i (\alpha_2-\beta_2)\sigma_3}$};
\node [above] at (-2.5,4) {$e^{\pi i (\alpha_1-\beta_1)\sigma_3}$};
\node [left] at (-1.5,7.5) 
{$\begin{pmatrix}1&0\\-1&1\end{pmatrix}
$};
\node [left] at (-1.8,8.8) 
{$\Gamma_2
$};
\node [right] at (1.5,7.5) 
{$\begin{pmatrix}1&1\\0&1\end{pmatrix}
$};
\node [right] at (1.8,8.8) 
{$\Gamma_1
$};
\node [left] at (-1.5,-3.5) 
{$\begin{pmatrix}1&0\\-1&1\end{pmatrix}
$};
\node [left] at (-1.8,-4.8) 
{$\Gamma_3
$};
\node [right] at (1.5,-3.5) 
{$\begin{pmatrix}1&1\\0&1\end{pmatrix}
$};
\node [right] at (1.9,-4.8) 
{$\Gamma_4
$};
\node[right] at (0,1) 
{$J_0
$};
\node[right] at (0,3.5) 
{$J_0
$};
\node[right] at (0,2.5) 
{$J_0
$};
\node[right] at (0,5) 
{$J_0
$};
\node[right] at (0,-1.3) 
{$J_0
$};
\node[right] at (3.5,-2) 
{$J_0=\begin{pmatrix}
0&1\\-1&1
\end{pmatrix}
$};
\node [below left] at (0,4) {$iw_1$};
\node [below left] at (0,0) {$w_4$};
\node [below left] at (0,3) {$iw_2$};
\node [below left] at (0,2) {$iw_3$};
\node [left] at (0,6) {$iu$};
\node [left] at (0,-2) {$-iu$};
\draw [fill=black] (0,0) circle[radius=0.07]; 
\draw [fill=black] (0,2) circle[radius=0.07];
\draw [fill=black] (0,3) circle[radius=0.07];
\draw [fill=black] (0,4) circle[radius=0.07];

\end{tikzpicture}
\caption{Jumps of $\Phi$ for $4$ singularities.}\label{JumpsPhi}
\end{figure}
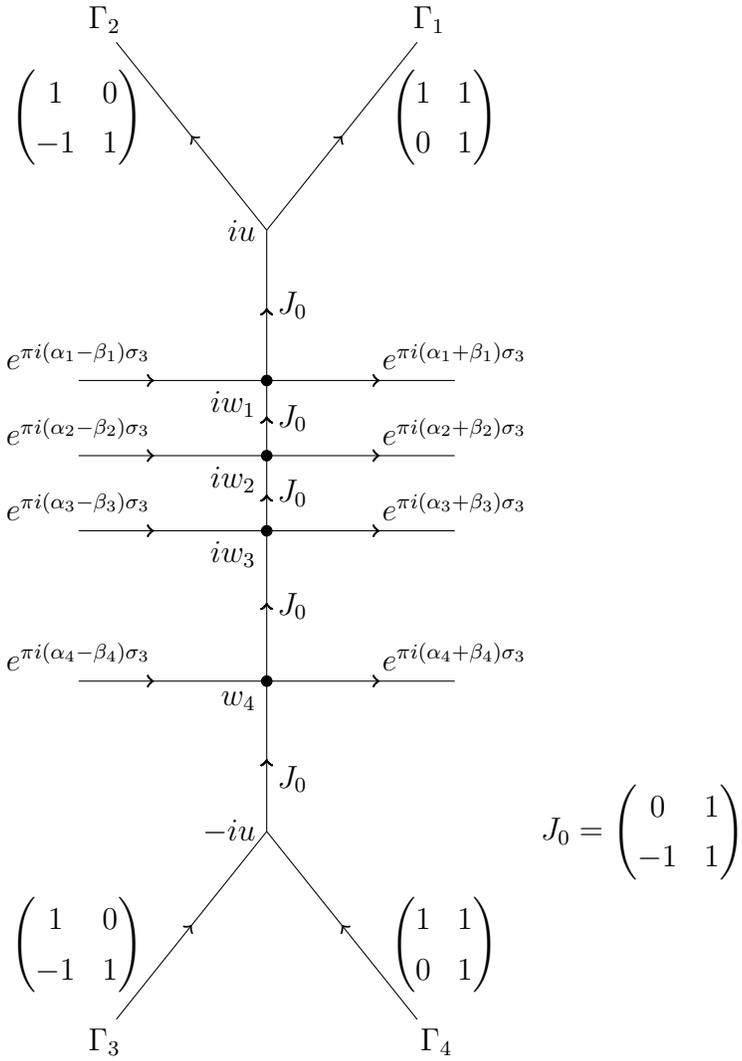
In this section we introduce and analyze a model Riemann--Hilbert problem, yielding results which we rely on to prove Propositions \ref{PropPoly} (b) and \ref{PropPoly2} in Section \ref{SectionAsym}.

We pose a Riemann--Hilbert problem for $\Phi=\Phi\left(\zeta;(w_j,\alpha_j,\beta_j)_{j=1}^{\mu}\right)$ with parameters 
\begin{itemize}
\item $\mu=1,2,3,\dots$,
\item $-u/2\leq w_\mu<w_{\mu-1}<\dots<w_2<w_1\leq u/2$, where $w_1\geq 0$,
\item $\alpha_j\geq0$ and $\Re \beta_j=0$, with $(\alpha_j,\beta_j)\neq (0,0)$ for $j=1,2,\dots,\mu$,\end{itemize}
where $u>0$ is some fixed constant.

The model RH problem will later be used to construct a local parametrix at each cluster of points, where $\mu$ will be the number of points in the cluster. In particular it means that the ordering of the $\alpha_j,\beta_j$ here do not necessarily correspond with those in the definition of the Toeplitz determinant, see Section \ref{SectionAsym} for details on how the model RH problem is utilized.

\subsubsection*{RH problem for $\Phi$}

\begin{itemize}
\item[(a)] $\Phi$ is analytic on $\mathbb C\setminus \Gamma_\Phi$, where $\Gamma_\Phi$ is described by Figure \ref{JumpsPhi} for $\mu=4$, and is in general given by
\begin{equation}\nonumber
\begin{aligned}
\Gamma_\Phi&=\cup_{j=0}^5 \Gamma_j , & \Gamma_0&=[-iu,iu],\\
\Gamma_1&=iu+e^{\pi i/4}\mathbb R_+,& \Gamma_2&=iu+e^{3\pi i /4}\mathbb R_+,
& \Gamma_3&=-iu+e^{5\pi i/4}\mathbb R_+\\
\Gamma_4&=-iu+e^{7\pi i /4}\mathbb R_+,& \Gamma_5&=\cup_{j=1}^\mu \{z:\Im z=w_j\},
\end{aligned}
\end{equation}
where $iu+e^{\pi i j/4}\mathbb R_+=\{z:\arg (z-iu)=\pi j/4\}$,
with the orientation of $\Gamma_5$ taken to the right,  and the orientation of $\Gamma_j$ taken upwards for $j=0,\dots, 4$.
\item[(b)] $\Phi$ has the following jumps on $\Sigma$:
\begin{equation}\nonumber
\begin{aligned}
\Phi_+(\zeta)&=\Phi_-(\zeta)e^{\pi i (\alpha_j+\beta_j)} && \textrm{for $\Im \zeta=w_j$ and $\Re \zeta>0$},\\
\Phi_+(\zeta)&=\Phi_-(\zeta)e^{\pi i (\alpha_j-\beta_j)} && \textrm{for $\Im \zeta=w_j$ and $\Re \zeta<0$},\\
\Phi_+(\zeta)&=\Phi_-(\zeta)\begin{pmatrix}
1&1\\0&1
\end{pmatrix}
 && \textrm{for }\zeta \in \Gamma_1,\Gamma_4,\\
 \Phi_+(\zeta)&=\Phi_-(\zeta)\begin{pmatrix}
1&0\\-1&1
\end{pmatrix}
 && \textrm{for }\zeta \in \Gamma_2,\Gamma_3,\\
 \Phi_+(\zeta)&=\Phi_-(\zeta)J_0
 && \textrm{for }\zeta \in \Gamma_0\setminus \{-iu,iu,iw_1,\dots,iw_\mu\},\\
\end{aligned}
\end{equation}
where $J_0=\begin{pmatrix}
0&1\\-1&1
\end{pmatrix}$.
\item[(c)] As $\zeta \to \infty$,
\begin{equation}\nonumber
\begin{aligned}
\Phi(\zeta)&=\left(I+\frac{\Phi_1}{\zeta}+\mathcal O(\zeta^{-2})\right)e^{-\frac{\zeta}{2} \sigma_3}\prod_{j=1}^\mu (\zeta-iw_j)^{-\beta_j\sigma_3}\exp\left[\pi i(\beta_j-\alpha_j)\chi_{w_j}(\zeta)\sigma_3\right], \\
\chi_w(\zeta)&=\begin{cases}
0&\textrm{for }\Im \zeta>w,\\
1&\textrm{for }\Im \zeta<w,\\
\end{cases}
\end{aligned}
\end{equation}
where the branches are chosen such that $\arg (\zeta-iw_j)\in [0,2\pi)$ for $j=1,\dots,\mu$, and where $\Phi_1=\Phi_1\left(\mu;(w_j,\alpha_j,\beta_j)_{j=1}^\mu\right)$ is independent of $\zeta$.
\item[(d)]  $\Phi(\zeta)$ is bounded as $\zeta \to \pm iu$. 
As $\zeta \to iw_j$ for $j=1,\dots, k$ in the sector $\arg (\zeta-iw_j)\in (\pi/2,\pi)$,
\begin{equation}\nonumber
\Phi(\zeta)=\begin{cases}F_j(\zeta)(\zeta-iw_j)^{\alpha_j\sigma_3}\begin{pmatrix}
1&g(\alpha_j,\beta_j)\\0&1
\end{pmatrix} &\textrm{for }2\alpha_j\notin \mathbb N=\{0,1,2,\dots\},\\
F_j(\zeta)(\zeta-iw_j)^{\alpha_j\sigma_3}\begin{pmatrix}
1&g(\alpha_j,\beta_j)\log(\zeta-iw_j)\\0&1
\end{pmatrix} &\textrm{for }2\alpha_j\in \mathbb N,
\end{cases}
\end{equation}
for some function $F_j$ which is analytic on a neighbourhood of $iw_j$, and
\begin{equation}\label{def:g}
g(\alpha,\beta)=\begin{cases}-\frac{e^{2\pi i \beta}-e^{-2\pi i \alpha}}{2i\sin 2\pi \alpha}&\textrm{for }2\alpha\notin \mathbb N,\\
\frac{i}{2\pi}\left((-1)^{2\alpha}e^{2\pi i \beta}-1\right)
&\textrm{for }2\alpha\in \mathbb N.
\end{cases}
\end{equation}
\end{itemize}
Note that $\pm iu$ are not special points, and therefore the values of $u$ are not particularly important. However, we present the RH problem in this manner for notational convenience and to make it clear that the local behaviour at each singularity can be presented in the same form, also for the top and bottom singularity.

We also note that $g$ was chosen such that the local behaviour of $\Psi$ at the point $iw_j$ is consistent with the jumps. 

\subsubsection{The case of a single singularity $\mu=1$}
When there is only one singularity $\mu=1$ and $w_1,u=0$,  the  RH problem for $\Phi$ (and equivalent versions of it) has been studied by many authors. It was first solved by Kuijlaars and Vanlessen in \cite{KV,Vanlessen} for  $\beta=0$ in terms of Bessel functions, and brought to the setting of determinants by Krasovsky in \cite{Krasovsky}. For $\alpha=0$ it was solved by Its and Krasovsky in \cite{IK}, and a  solution for general $\alpha,\beta$ was found in terms of confluent hypergeometric functions by Deift, Its, Krasovsky in \cite{DIK, DIK3} and Moreno in \cite{Moreno}. 
 
 Claeys, Its and Krasovsky \cite{CIK} brought the above solution to the form which we will refer to. In \cite{CIK} the RH problem is  denoted by $M$, which we will denote by $M_{\textrm{CIK}}$, and by comparison of RH problems it follows that
\begin{equation}\nonumber
M_{\textrm{CIK}}(\zeta)=e^{\frac{\pi i}{2}(\alpha-\beta)\sigma_3}
\Phi(\zeta;\mu=1,w=0,\alpha,\beta)e^{\pi i (\alpha-\beta)\chi_0(\zeta)\sigma_3}e^{\frac{\pi i }{2}(\beta-\alpha)\sigma_3},
\end{equation}
when one takes $u=0$.
The solution to $M_{\textrm{CIK}}$ may also be found in \cite{CK}, Section 4, where we find the following formula
\begin{equation}\label{1singPhi}
\Phi_1(\mu=1=,w=0,\alpha,\beta)=\begin{pmatrix}
\alpha^2-\beta^2&-e^{-\pi i (\alpha+\beta)}\frac{\Gamma(1+\alpha-\beta)}{\Gamma(\alpha+\beta)}\\
e^{\pi i (\alpha+\beta)}\frac{\Gamma(1+\alpha+\beta)}{\Gamma(\alpha-\beta)}
&\beta^2-\alpha^2
\end{pmatrix},
\end{equation}
where $\Gamma$ is Euler's $\Gamma$ function.

\subsubsection{The case of multiple singularities $\mu>1$}
In the case of 2 singularities $\mu=2$,  an equivalent version of the RH problem for $\Phi$ was proven to have a unique solution by Claeys and Krasovsky \cite{CK}, and to be connected to the Painlev\'{e} V equation. See also \cite{Van2} for a reference on Riemann--Hilbert problems connected to the Painlev\'{e} equations.
 The proof of a unique solution by \cite{CK} generalizes easily to our  situation of $\mu=1,2,3,\dots$ singularities, and we have included a proof of the following proposition in the Appendix for the reader's convenience.
\begin{proposition}\label{PropSolv} Let $\alpha_j\geq 0$ and $\Re \beta_j=0$ for $j=1,2,\dots\mu$.
There exists a unique solution to the Riemann--Hilbert problem for $\Phi$.
\end{proposition}

\subsection{Continuity of $\Phi$ for varying $w_j$'s}\label{SecCon}
The main result of Section \ref{ModelProblem} is the following.

\begin{lemma}\label{Lemmacont}
Let $\alpha_j\geq 0$ and $\Re \beta_j=0$ for $j=1,2,\dots, \mu$. Then the following two statements hold.

(a) Given $u>0$,
\begin{equation}\label{Phiinfunif}
\Phi(\zeta)e^{\frac{\zeta}{2} \sigma_3}\prod_{j=1}^\mu (\zeta-iw_j)^{\beta_j\sigma_3}\exp\left[\pi i(-\beta_j+\alpha_j)\chi_{w_j}(\zeta)\sigma_3\right]
=I+\mathcal O \left( \frac{1}{\zeta}\right),
\end{equation}
  as $\zeta \to \infty$, uniformly for $-u/2\leq w_\mu<\dots<w_1\leq u/2$. 

(b)
As $w_1-w_\mu\to 0$,
\begin{equation} \label{Phi1limsing}
\Phi_1\left(\mu;(w_j,\alpha_j,\beta_j)_{j=1}^\mu\right)=\begin{pmatrix}
\mathcal A^2-\mathcal B^2&-e^{-\pi i (\mathcal A+\mathcal B)}\frac{\Gamma(1+\mathcal A-\mathcal B)}{\Gamma(\mathcal A+\mathcal B)}\\
e^{\pi i (\mathcal A+\mathcal B)}\frac{\Gamma(1+\mathcal A+\mathcal B)}{\Gamma(\mathcal A-\mathcal B)}
&\mathcal B^2-\mathcal A^2
\end{pmatrix}+\mathcal O(w_1-w_\mu),
\end{equation} 
where $\mathcal A=\sum_{j=1}^\mu \alpha_j$ and $\mathcal B=\sum_{j=1}^\mu \beta_j$.
\end{lemma}

The first step in the proof of Lemma \ref{Lemmacont} is to transform the RH problem for $\Phi$ to a RH problem for $\widehat \Phi$ which is analytic except on the imaginary axis $\Re z=0$, and in particular the jump contour is independent of the locations of the singularities $w_j$ (though the jumps themselves will vary with the location of the singularities).

\subsection{Transformation of RH problem}
Let 
\begin{equation}\nonumber
\Pi(z)=\begin{cases}
I&{\rm for}\, z\in I,IV,\\
\begin{pmatrix}
1&-1\\0&1
\end{pmatrix}&{\rm for}\, z\in II,VI,\\
\begin{pmatrix}
1&0\\-1&1
\end{pmatrix}
&{\rm for}\, z\in III,V,
\end{cases}
\end{equation}
where I-VI are regions in the complex plane given in Figure \ref{Regions}.

\begin{figure}
\begin{tikzpicture}
\draw  (0,-3)--(0,7);
\draw  (-2,-2.5)--(0,0);
\draw (2,-2.5)--(0,0);
\draw  (0,4)-- (-2,6.5);
\draw  (0,4)-- (2,6.5);

\node [left] at (-1.8,6.8) 
{$\Gamma_2
$};
\node [right] at (1.8,6.8) 
{$\Gamma_1
$};
\node [left] at (-1.8,-2.8) 
{$\Gamma_3
$};
\node [right] at (1.9,-2.8) 
{$\Gamma_4
$};

\node [left] at (0,4) {$iu$};
\node [left] at (0,0) {$-iu$};

\node at (2,2){I};
\node at (0.8,6) {II};
\node at (-0.8,6) {III};
\node at (-2,2){IV};
\node at (-0.8,-2) {V};
\node at (0.8,-2) {VI};

\draw [decoration={markings, mark=at position 0.1 with {\arrow[thick]{>}}},
        postaction={decorate}]
        [decoration={markings, mark=at position 0.4 with {\arrow[thick]{>}}},
        postaction={decorate}]
        [decoration={markings, mark=at position 0.6 with {\arrow[thick]{>}}},
        postaction={decorate}]
	[decoration={markings, mark=at position 0.9 with {\arrow[thick]{>}}},
        postaction={decorate}]  (9,-3)--(9,7);



\node [left] at (9,4) {$iw_1$};
\node [left] at (9,2) {$iw_2$};
\node [left] at (9,0) {$iw_3$};

\draw [fill=black] (9,4) circle[radius=0.07]; 
\draw [fill=black] (9,2) circle[radius=0.07]; 
\draw [fill=black] (9,0) circle[radius=0.07]; 

\node [right] at (9,5) {$J_0$};
\node [right] at (9,3) {$J_1$};
\node [right] at (9,1) {$J_2$};
\node [right] at (9,-1) {$J_3$};

\end{tikzpicture}
\caption{On the left we display Regions I-VI. On the right we display the jumps of RH problem for $\widehat \Phi$ when $\mu=3$, which are all on the imaginary axis $\Re z=0$, and we have not displayed $\pm iu$ because these points are not relevant for $\widehat \Phi$.}\label{Regions}
\end{figure}
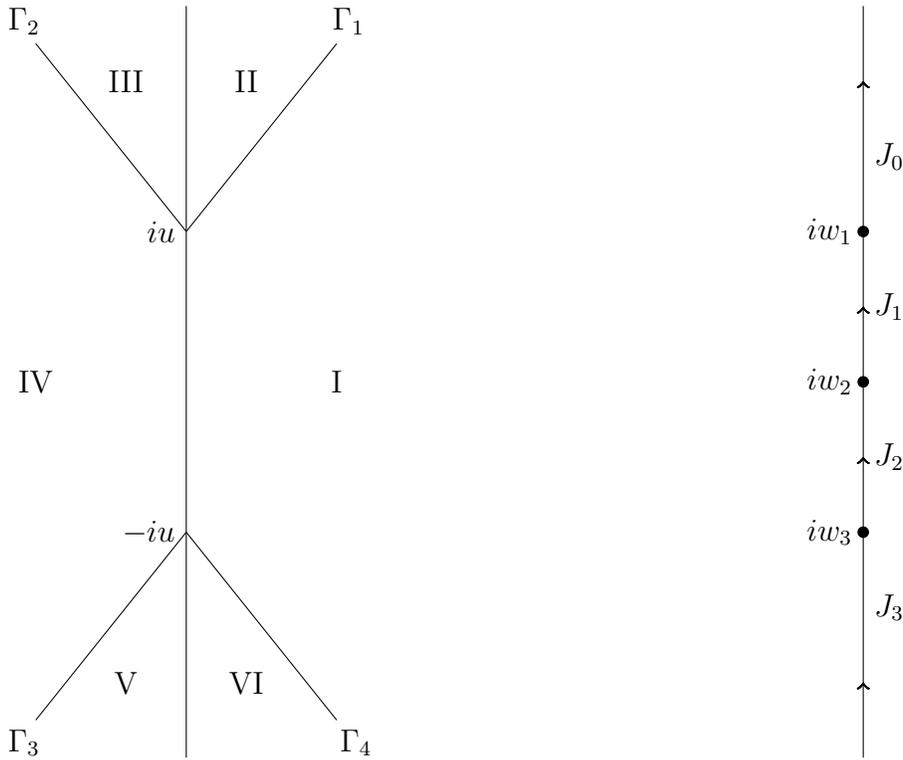
Let $\widehat \Phi$ be defined in terms of $\Phi$ as follows.
\begin{equation}\label{defhatPhi}
\widehat \Phi(\zeta)=\begin{cases}\Phi(\zeta)\Pi(\zeta)\prod_{j=1}^\mu \exp\left[\pi i (\alpha_j-\beta_j)\chi_{w_j}(\zeta)\sigma_3\right]&\textrm{for } \Re \zeta <0,\\
\Phi(\zeta)\Pi(\zeta)\prod_{j=1}^\mu \exp\left[\pi i (\alpha_j+\beta_j)\chi_{w_j}(\zeta)\sigma_3\right]
&\textrm{for } \Re \zeta >0.
\end{cases}
\end{equation}

Then $\widehat \Phi$ solves the following RH problem.

\subsubsection*{RH problem for $\widehat \Phi$}
\begin{itemize}
\item[(a)] $\widehat \Phi$ is analytic on $\mathbb C\setminus (-i\infty,i\infty)$, 
with the orientiation of $(-i\infty,i\infty)$ upwards.
\item[(b)]Let $w_0=+\infty$ and $w_{\mu+1}=-\infty$. On $(iw_{j+1},iw_j)$,
\begin{equation}\nonumber
\widehat \Phi_+=\widehat \Phi_-J_j,
\end{equation}
for $j=0,1,\dots,\mu$,
where $J_0$ was defined in condition (b) of the RH problem for $\Phi$, and 
\begin{equation}\nonumber
J_j=\begin{pmatrix}
0&
\exp \left[-2\pi i \sum_{\nu=1}^j \alpha_\nu\right]\\
-\exp \left[2\pi i \sum_{\nu=1}^j \alpha_\nu\right]&
\exp \left[2\pi i \sum_{\nu=1}^j \beta_\nu\right]
\end{pmatrix}
\end{equation}
for $j=1,2,\dots, \mu$.

\item[(c)] The behaviour of $\widehat \Phi(\zeta)$ as $\zeta\to \infty$  is inherited  from conditions (c) of the RH problem for $\Phi$.
\item[(d)]The behaviour of $\widehat \Phi(\zeta)$ in neighbourhoods of $iw_j$, $j=1,\dots,\mu$,  is inherited  from conditions (d) of the RH problem for $\Phi$.

\end{itemize}

\subsection{Steepest descent analysis of $\widehat \Phi$}

We now prove that $\widehat \Phi$ is continuous with respect to the parameters $w_1,\dots,w_\mu$.
 
 Partition $\{w_1,\dots,w_\mu\}$ into $\tau$ disjoint sets 
\begin{equation} \widetilde C_j(\epsilon)=\{w_\nu^{(j)}\}_{\nu=1}^{M_j},\qquad j=1,2,\dots,\tau,
\end{equation}
such that
\begin{equation}
|w_\nu^{(j)}-\mathcal W_j|<\epsilon, \qquad \textrm{for all }\nu=1,2,\dots,M_j,
\end{equation}
where $\mathcal W_1,\dots,W_\tau$ are distinct fixed points. We will consider the asymptotics of $\widehat \Phi(\zeta)$ as $\epsilon \to 0$. We order the points so that
\begin{equation}
w_1^{(j)}>\dots >w_{M_j}^{(j)}, \qquad   \mathcal W_1>\dots>\mathcal W_\tau.
\end{equation}

Denote 
\begin{equation}\nonumber
\mathcal A^{(j)}=\sum_{\nu=1}^{M_j}\alpha_\nu^{(j)}, \qquad \mathcal B^{(j)}=\sum_{\nu=1}^{M_j}\beta_\nu^{(j)}, 
\end{equation}
for $j=1,2,\dots, \tau$. 

We plan to approximate the RH problem associated with the $w_j$'s by the RH problem associated with the $\mathcal W_j$'s, and so for increased clarity we label them as different functions.

Let the RH problem associated with $w_1,\dots,w_{\mu}$ be denoted by 
\begin{equation}\Psi(\zeta)=\widehat \Phi\left(\zeta;\mu ,(w_i,\alpha_i,\beta_i)_{i=1}^{\mu}\right),\end{equation} and the RH problem associated with $\mathcal W_1,\dots, \mathcal W_\tau$ by \begin{equation}N(\zeta)=\widehat \Phi\left(\zeta;\tau,(\mathcal W_j,\mathcal A^{(j)},\mathcal B^{(j)})_{j=1}^\tau\right).\end{equation}

We note that $N$ has the same jumps as $\Psi$ except on neighbourhoods containing $\mathcal W_1,\dots,\mathcal W_\tau,$ and that $\Psi(\zeta)N(\zeta)^{-1}\to I$ as $\zeta\to \infty$ by condition (c) for the RH problem for $\Phi$ and the definition of $\widehat \Phi$. 

We will additionally need to show that there exists a local parametrix $Q(\zeta)$ on fixed neighbourhoods $U_{\mathcal W_j}$ of $\mathcal W_j$, such that $\Psi(\zeta)Q(\zeta)^{-1}$ is analytic on $U_{\mathcal W_j}$, and 
\begin{equation}
Q(\zeta)N^{-1}(\zeta)=I+\mathcal O(\epsilon),
\end{equation}
as $\epsilon\to 0$, uniformly for $\zeta\in \partial U_{\mathcal W_j}$. Although we only need existence of such a local parametrix, we prove the existence by construction, and we do this in the next subsection, Section \ref{LocalModel}. By standard theory of small norm problems, see e.g. \cite{Deift}, it will follow that $N$ approximates $\Psi$ well outside of the neighbourhoods $\cup_{j=1}^\tau U_{\mathcal W_j}$, which we will subsequently use in Section \ref{ProofCont} to prove Lemma \ref{Lemmacont}.

\subsubsection{Local parametrix}\label{LocalModel}
We construct a local parametrix at the point $\mathcal W_j$ which will contain the points $w_1^{(j)},\dots,w_{M_j}^{(j)}$, for $j=1,\dots, \tau$, and are inspired here by a similar construction in \cite{CK} in the special case of two singularities.

Throughout the section \ref{LocalModel}, $j$ will be fixed, and to reduce the number of superscripts, we denote throughout the section
\begin{equation}
y_\nu=w_\nu^{(j)}, \qquad \tilde \alpha_\nu=\alpha_\nu^{(j)},\qquad \tilde \beta_\nu= \beta_\nu^{(j)}, 
\end{equation}
for $\nu=1,2,\dots,M_j$.
Let $U_{\mathcal W_j}$ be a fixed open disc centered at $i\mathcal W_j$ with a fixed radius $R>0$.

We first take a transformation $\Psi\to \Psi_j$, where $\Psi_j$ is analytic for all 
\begin{equation} \left \{\zeta \in U_{\mathcal W_j}:\arg(\zeta-iy_1)\neq 3\pi/2\right \},\end{equation} and similarly a transformation $N\to N_j$ such that $N_j$ is analytic for all  
\begin{equation} \left \{ \zeta\in U_{\mathcal W_j}: \arg\left(\zeta-i\mathcal W_j\right) \neq 3\pi/2\right \}.\end{equation}
On $U_{\mathcal W_j}$, we define
\begin{equation}\nonumber
\Psi_j(\zeta)=\begin{cases}
\Psi(\zeta)J_{{\rm L}}^{(j)} &\textrm{for } \Re \zeta<0,\\
\Psi(\zeta)J_{{\rm R}}^{(j)} &\textrm{for } \Re \zeta>0,
\end{cases}\qquad 
N_j(\zeta)=\begin{cases}
N(\zeta)J_{{\rm L}}^{(j)} &\textrm{for } \Re \zeta<0,\\
N(\zeta)J_{{\rm R}}^{(j)} &\textrm{for } \Re \zeta>0,
\end{cases}
\end{equation}
where 
\begin{equation}\nonumber
J_{{\rm L}}^{(j)}=\prod_{s=1}^{j-1} \exp\left[-\pi i\left(\mathcal A^{(s)}-\mathcal B^{(s)}\right) \sigma_3 \right], \qquad J_{{\rm R}}^{(j)}=\left(\prod_{s=1}^{j-1} \exp\left[-\pi i\left(\mathcal A^{(s)}+\mathcal B^{(s)}\right) \sigma_3 \right]\right)J_0,
\end{equation}
for $j=2,3,\dots,\tau$, and 
\begin{equation}\nonumber
J_L^{(1)}=I, \qquad J_R^{(1)}=J_0. \end{equation}
Denote
\begin{equation}\nonumber
\mathcal A_\nu=\mathcal A_\nu^{(j)}=\sum_{s=1}^{\nu}\tilde \alpha_s, \qquad \mathcal B_\nu=\mathcal B_\nu^{(j)}=\sum_{s=1}^{\nu}\tilde \beta_s, \qquad \nu=1,\dots,M_j.
\end{equation}
  On $(iy_{\nu+1},i y_\nu)$, $\Psi_j$ has the jumps
\begin{equation}\label{jumpPsi1}
\begin{aligned}
\Psi_{j,+}(\zeta)&=\Psi_{j,-}(\zeta)J_0^{-1}\exp\left[-\pi i \left(\mathcal A_\nu+
\mathcal B_\nu \right)\sigma_3\right] J_0
\exp\left[\pi i \left(\mathcal A_\nu-\mathcal B_\nu\right)\sigma_3\right]\\
&=\Psi_{j,-}(\zeta)\begin{pmatrix}
\exp \left(2\pi i \mathcal A_\nu\right)
&\exp \left(-2\pi i \mathcal A_\nu\right)-\exp \left(2\pi i \mathcal B_\nu\right)\\ 0 &\exp \left(-2\pi i \mathcal A_\nu \right),
\end{pmatrix},
\end{aligned}
\end{equation}
for $\nu=1,\dots,M_{j-1}$ where orientation of the contour is taken upwards,
where $J_0$ was defined in condition (b) of the RH problem for $\Phi$, and on $\arg (\zeta-y_{M_j})=3\pi/2$,
\begin{equation}\label{jumpPsi2}
\Psi_{j,+}(\zeta)=\Psi_{j,-}(\zeta)
\begin{pmatrix}
\exp \left(2\pi i \mathcal A_{M_j}\right)
&\exp \left(-2\pi i \mathcal A_{M_j}\right)-\exp \left(2\pi i \mathcal B_{M_j}\right)\\ 0 &\exp \left(-2\pi i \mathcal A_{M_j}\right)
\end{pmatrix}.
\end{equation}

We search for a local parametrix $Q_j$ such that $Q_j$ has the same jumps as $\Psi_j$ on $U_{\mathcal W_j}$ and such that $Q_jN_j^{-1}=I+\mathcal O(\epsilon)$ as $\epsilon \to 0$, uniformly on the boundary $\partial U_{\mathcal W_j}$ 
for $j=1,2,\dots, \tau$ (it follows that $Q_j$ differs from $Q$ by right multiplication of $J_{{\rm L}}^{(j)}$ and $J_{{\rm R}}^{(j)}$ for $\Re z$ negative and positive respectively). The approach depends on whether or not $2\mathcal A_{M_j}\in \mathbb N$.

\subsubsection*{Local parametrix for $2\mathcal A_{M_j}\notin \mathbb N$}\label{SectionLocalPsi}
Assume that $2\mathcal A_{M_j}\notin \mathbb N$, and define
\begin{equation}\label{defQj}
\begin{aligned}
Q_j(\zeta)&=E_j(\zeta) \widehat Q_j(\zeta)
\prod_{\nu=1}^{M_j}\left(\zeta-iy_\nu\right)^{\tilde \alpha_\nu\sigma_3}\begin{pmatrix}
1&g\left(\mathcal A_{M_j},\mathcal B_{M_j}\right)\\0&1
\end{pmatrix},\\
\widehat Q_j(\zeta)&=\begin{pmatrix}
1&\sum_{\nu=1}^{M_j-1} c_\nu
\int_{iy_{\nu+1}}^{iy_\nu} \prod_{s=1}^{M_j}\left|\lambda-iy_s\right|^{2\tilde \alpha_s} \frac{d\lambda}{\pi i(\lambda-\zeta)}\\
0&1
\end{pmatrix}
\end{aligned}
\end{equation}
where $g$ was defined in \eqref{def:g},  $E_j$ is an analytic function  given below in  \eqref{defEj},  $\arg (\zeta-iy_\nu)\in (-\pi/2,3\pi/2)$,
and
\begin{equation}\nonumber
c_\nu=c_\nu^{(j)}=\begin{cases}i\sin\left(2\pi \mathcal A_\nu\right) e^{\pi i \mathcal A_{M_j}}\left(g\left(\mathcal A_\nu,\mathcal B_\nu\right)-g\left(\mathcal A_{M_j},\mathcal B_{M_j}\right)\right) & \rm{for} \,2\mathcal A_\nu \notin \mathbb N,\\
\exp\left[\pi i \mathcal A_{M_j}\right]\left(\exp\left[-2\pi i \mathcal A_\nu\right]-\exp \left[2\pi i\mathcal B_\nu\right]\right)/2
& \rm{for} \,2\mathcal A_\nu \in \mathbb N.
\end{cases}
\end{equation}
 
 We first consider the jumps of $Q_j$.
 If 
\begin{equation}\nonumber
F(\zeta)=\frac{1}{\pi i}\int_{ia}^{ib}h(\lambda)\frac{d\lambda}{\lambda-\zeta},
\end{equation}
for $a<b$ and $h$ is $L^{2}$ integrable on $[ia,ib]$, then $F$ is analytic on $\mathbb C\setminus [ia,ib]$. It follows that $\widehat Q_j$ is analytic on $U_{\mathcal W_j}\setminus [iy_{M_j},iy_1]$, and it 
 is easily verified by comparison with \eqref{jumpPsi2} that $Q_{j,-}^{-1}Q_{j,+}=\Psi_{j,-}^{-1}\Psi_{j,+}$ for  $\zeta$ with $\arg \left(\zeta-iy_{M_j}\right)=3\pi/2$. 
If in addition $h(\lambda)$ extends to an analytic function on a an open set containing $(ia,ib)$, then
\begin{equation}\nonumber
F_+(\zeta)=F_-(\zeta)+2h(\zeta),
\end{equation}
for all $\zeta \in (ia,ib)$, 
 with upward orientation,  so
\begin{equation}\nonumber
\widehat Q_j(\zeta)_+=\widehat Q_j(\zeta)_-\begin{pmatrix}
1&2c_\nu
 \prod_{s=1}^{M_j}\left|\zeta-iy_s\right|^{2\tilde \alpha_s} \\
0&1
\end{pmatrix}
\end{equation} 
 for $\zeta\in (iy_{\nu+1},iy_{\nu})$, $\nu=1,2,\dots, M_j-1$, and it follows that
 \begin{equation}\nonumber
  Q_j(\zeta)_+= Q_j(\zeta)_-\begin{pmatrix}
1&-g\left(\mathcal A_{M_j},\mathcal B_{M_j}\right)\\0&1
\end{pmatrix}\begin{pmatrix}
e^{2\pi i \mathcal A_\nu}  
  &2c_\nu e^{-\pi i \mathcal A_\nu}\\0
  &  e^{-2\pi i \mathcal A_\nu}  
 \end{pmatrix}
 \begin{pmatrix}
1&g\left(\mathcal A_{M_j},\mathcal B_{M_j}\right)\\0&1
\end{pmatrix}.
 \end{equation}
 By comparison with \eqref{jumpPsi1} and the definition of $c_\nu$, it follows that $Q_{j,-}^{-1}Q_{j,+}=\Psi_{j,-}^{-1}\Psi_{j,+}$ on $(iy_{\nu+1},iy_{\nu})$, $\nu=1,2,\dots, M_j-1$.  Furthermore, since $\widehat Q_j$ is bounded on $\mathcal U_j$ and by the fact that (using the definition of $\Psi$ and condition (d) of the RH problem for $\Phi$)
 \begin{equation} \nonumber
 \Psi(\zeta) \begin{pmatrix}1&-g(\mathcal A_{M_j},\mathcal B_{M_j})
 \\0&1\end{pmatrix}\prod_\nu^{M_j}(\zeta-iy_\nu)^{-2\tilde \alpha_\nu \sigma_3}\end{equation}
 is bounded on $\mathcal U_j$, and the fact that $E_j$ is analytic,  it follows that $\Psi_j(\zeta)Q_j(\zeta)^{-1}$ is analytic on $\mathcal U_j$.
  
  We define $E_j$ by 
\begin{equation}\label{defEj}
E_j(\zeta)=N_j(\zeta)\begin{pmatrix}
1&-g\left(\mathcal A_{M_j},\mathcal B_{M_j}\right)\\0&1
\end{pmatrix}\left(\zeta-i\mathcal W_j\right)^{-\mathcal A_{M_j}\sigma_3}.
\end{equation}
We recall that $\mathcal A_{M_j}=\mathcal A^{(j)}$.
Thus, by the definition of $N_j$ and condition (d) for the RH problem for $\Phi$, it follows that $N_j(\zeta)=\Phi\left(\zeta;\left(\mathcal W_i,\mathcal A^{(i)},\mathcal B^{(i)}\right)_{i=1}^\tau\right)$ for $\arg (\zeta)\in (\pi/2,\pi)$, and one verifies that the singularity of $N_j$ cancels with that of $(\zeta-i\mathcal W_j)^{-\mathcal A_{M_j}\sigma_3}$. It is easily seen that $E_j$ has no jumps on $U_{\mathcal W_j}$, and thus it is analytic. 
For $\zeta \in \partial U_{\mathcal W_j}$, define
\begin{equation}\nonumber
\Delta_j(\zeta)=E_j(\zeta)\widehat Q_j(\zeta)\prod_{\nu=1}^{M_j}\left(\zeta-iy_\nu \right)^{\tilde \alpha_{\nu}\sigma_3}\left(\zeta-i\mathcal W_j\right)^{-\mathcal A_{M_j}\sigma_3}E_j^{-1}(\zeta).\end{equation}
We first note that $E_j$ are analytic functions on $U_{\mathcal W_j}$, and since they are independent of $\epsilon$, the are uniformly bounded on $\partial U_{\mathcal W_j}$.
Since
\begin{equation}\nonumber
\int_{iy_{\nu+1}}^{iy_\nu} \prod_{j=1}^{M_j}\left|\lambda-iy_s\right|^{2\tilde \alpha_s} \frac{d\lambda}{\lambda-\zeta}=\mathcal O(\epsilon), \qquad  \frac{\prod_{\nu=1}^{M_j}\left(\zeta-iy_{\nu}\right)^{\tilde \alpha_\nu}}{\left(\zeta-i\mathcal W_j\right)^{\mathcal A_{M_j}}}=\mathcal O(\epsilon)
\end{equation}
as $\epsilon \to 0$, uniformly for $\zeta\in \partial U_{\mathcal W_j}$, it follows by \eqref{defQj} and the boundedness of $E_j$ on $\partial U_{\mathcal W_j}$, that
\begin{equation}\label{small1}
\Delta_j(\zeta)=I+\mathcal O(\epsilon)
\end{equation}
as $\epsilon \to 0$, uniformly for $\zeta\in \partial U_{\mathcal W_j}$.

\subsubsection*{Local parametrix for $2\mathcal A_{M_j}\in \mathbb N$}
Assume that $2\mathcal A_{M_j}\in \mathbb N$, and define
\begin{equation}\nonumber
\begin{aligned}
Q_j(\zeta)&=E_j(\zeta) \widehat Q_j(\zeta)
\prod_{\nu=1}^{M_j}\left(\zeta-iy_\nu\right)^{\tilde\alpha_\nu\sigma_3}\begin{pmatrix}
1&g\left(\mathcal A_{M_j},\mathcal B_{M_j}\right)\log\left( \zeta-iy_{M_j}\right)\\0&1
\end{pmatrix},\\
\widehat Q_j(\zeta)&=\begin{pmatrix}
1&\sum_{\nu=1}^{M_j-1} 
\int_{iy_{\nu+1}}^{iy_\nu} 
\left( d_\nu+e_\nu\log \left(\lambda-iy_{M_j}\right)\right) \prod_{s=1}^{M_j}\left|\lambda-iy_s\right|^{2\tilde \alpha_s} \frac{d\lambda}{\pi i(\lambda-\zeta)}\\
0&1
\end{pmatrix}
\end{aligned}
\end{equation}
where $g$ was defined in \eqref{def:g},  $E_j$ is an analytic function  given below in  \eqref{defEj2}, the argument $\arg \left(\zeta-iy_\nu\right)\in (-\pi/2,3\pi/2)$, and
\begin{equation}\nonumber
\begin{aligned}
d_\nu=d_\nu^{(j)}&=\frac{e^{\pi i \mathcal A_{M_j} }}{2}\left(
e^{-2\pi i \mathcal A_\nu}
-e^{2\pi i \mathcal B_\nu}
\right),
\\
e_\nu=e_\nu^{(j)}&=
-ie^{\pi i \mathcal A_{M_j}}
g\left(\mathcal A_{M_j},\mathcal B_{M_j}\right)
\sin 2\pi \mathcal A_\nu
.
\end{aligned}
\end{equation}
As in the case for $2\mathcal A_{M_j}\notin \mathbb N$, it is easily seen that $\widehat Q_j$ is analytic on $U_{\mathcal W_j}\setminus\left[iy_{M_j},iy_1\right]$, and it follows that $Q_{j,-}^{-1}Q_{j,+}=\Psi_{j,-}^{-1}\Psi_{j,+}$ on $\arg(\zeta-iy_{M_j})=3\pi/2$.
On  $ (iy_{\nu+1},iy_{\nu})$, $\nu=1,2,\dots, M_j-1$,
\begin{equation}\nonumber
\widehat Q_j(\zeta)_+=\widehat Q_j(\zeta)_-\begin{pmatrix}
1&2\left( d_\nu+e_\nu \log \left(\zeta-iy_{M_j}\right)\right)
 \prod_{s=1}^{M_j}\left|\zeta-iy_s\right|^{2\tilde \alpha_s} \\
0&1
\end{pmatrix},
\end{equation} 
  and it follows that
 \begin{multline}\nonumber 
  Q_j(\zeta)_+= Q_j(\zeta)_-\begin{pmatrix}
1&-g\left(\mathcal A_{M_j},\mathcal B_{M_j}\right)\log \left(\zeta-iy_{M_j}\right)\\0&1
\end{pmatrix}\\ \times 
\begin{pmatrix}
e^{2\pi i \mathcal A_\nu}  
  &2\left( d_\nu+e_\nu\log \left(\zeta-iy_{M_j}\right)\right)e^{-\pi i \mathcal A_{M_j} }\\0
  &  e^{-2\pi i \mathcal A_\nu}  
 \end{pmatrix} 
\begin{pmatrix}
1&g\left(\mathcal A_{M_j},\mathcal B_{M_j}\right)\log \left(\zeta-iy_{M_j}\right) \\0&1
\end{pmatrix}.
 \end{multline}
 By comparison with \eqref{jumpPsi1} and the definition of $d_\nu$, $e_\nu$, it follows that $Q_{j,-}^{-1}Q_{j,+}=\Psi_{j,-}^{-1}\Psi_{j,+}$ on $(iy_{\nu+1},iy_{\nu})$, $\nu=1,2,\dots, M_j-1$.

 We define $E_j$ by 
\begin{equation}\label{defEj2}
E_j(\zeta)=N_j(\zeta)\begin{pmatrix}
1&-g\left(\mathcal A_{M_j},\mathcal B_{M_j}\right)\log(\zeta-i\mathcal W_j)\\0&1
\end{pmatrix}\left(\zeta-i\mathcal W_j\right)^{-\mathcal A_{M_j}},
\end{equation}
and in a similar manner to the case $2\mathcal A_{M_j}\notin \mathbb N$,
it follows that $E_j$ is analytic on $U_{\mathcal W_j}$, and that
\begin{equation}\label{small2}
Q_j(\zeta)N_j(\zeta)^{-1}=I+\mathcal O(\epsilon)
\end{equation}
as $\epsilon \to 0$, uniformly for $\zeta\in \partial U_{\mathcal W_j}$.

\subsubsection{Small norm matrix}\label{SmallnormModel}

Define $Q$ on $\cup_{j=1}^{\tau}U_{\mathcal W_j}$ by
\begin{equation}\nonumber
Q(\zeta)=\begin{cases}
Q_j(\zeta)\left(J_{{\rm L}}^{(j)}\right)^{-1} &\textrm{for } \Re \zeta<0,\\
Q_j(\zeta)\left(J_{{\rm R}}^{(j)}\right)^{-1} &\textrm{for } \Re \zeta>0.
\end{cases}
\end{equation}

Let $ \widehat R$ be given by
\begin{equation}\nonumber
 \widehat R(\zeta)=\begin{cases}
 \Psi(\zeta)N(\zeta)^{-1} &\textrm{for } z\in \mathbb C \setminus  \left(\cup_{j=1}^r  U_{\mathcal W_j}\right), \\
 \Psi(\zeta)Q(\zeta)^{-1} &\textrm{for } z\in  U_{\mathcal W_j},\, j=1,\dots,r.
\end{cases}
\end{equation}
By \eqref{small1} and \eqref{small2}, it follows that
 $\widehat R$ satsisfies the following RH problem.
\subsubsection*{RH problem for $\widehat R$}
\begin{itemize}
\item[(a)]
Then $\widehat R$ is analytic on $\mathbb C\setminus \cup_{j=1}^\tau \partial U_{\mathcal W_j}$.
\item[(b)]On $\cup_{j=1}^\tau \partial U_{\mathcal W_j}$,
 \begin{equation}\nonumber
\widehat R_+(\zeta)=\widehat R_-(\zeta)(I+\mathcal O(\epsilon)),
\end{equation}  
as $\epsilon \to 0$, uniformly on $\cup_{j=1}^\tau U_{\mathcal W_j}$.

\item[(c)]
 $\widehat R(\zeta)=I+\mathcal O(\zeta^{-1})$ as $z\to \infty$.
\end{itemize}
By standard small norm analysis,
\begin{equation}\label{hatRsmall}
\widehat R(\zeta)=I+\mathcal O\left(\frac{\epsilon}{(|\zeta|+1)}\right),
\end{equation}
as $\epsilon \to 0$ uniformly for the parameters 
\begin{equation}\nonumber
|w_\nu^{(j)}-\mathcal W_j|< \epsilon,
\end{equation}
for fixed $\mathcal W_1,\dots,\mathcal W_\tau$, with the implicit constant depending only on $u$, and the parameters $\alpha_j,\beta_j$.

\subsection{Proof of Lemma \ref{Lemmacont}} \label{ProofCont}
We prove (a) by contradiction. Denote the left hand side of \eqref{Phiinfunif} by $F_{\Phi}(\zeta;w_1,\dots,w_\mu)$. Assume that there is a sequence of points $-u/2\leq w_\mu(k)<\dots<w_1(k)\leq u/2$ for $k=1,2,\dots $ and corresponding $\zeta_k$ such that $\zeta_k\to \infty$ as $k\to \infty$, satisfying
\begin{equation}\label{Contra}
|\zeta_k|\left|F_{\Phi}\left(\zeta_k;w_1(k),\dots,w_\mu(k)\right)-I\right|\to \infty,
\end{equation}
as $k\to \infty$. Then there would be a subsequence $k_i$ such that $w_j(k_i)\to w_j$ for $j=1,2,\dots,\mu$, for some points 
\begin{equation}\nonumber
-u/2\leq w_\mu\leq w_{\mu-1}\leq \dots \leq w_1\leq u/2.\end{equation}
We denote $\{w_1,\dots,w_\mu\}=\{\mathcal W_1,\dots,\mathcal W_\tau\}$, where the points $\mathcal W_\tau<\dots< \mathcal W_1$ are distinct.
Let 
\begin{equation}\nonumber
\epsilon_{k_i}=\max_{j=1,\dots, \mu} \min_{s=1,\dots,\tau}\left|w_j(k_i)-\mathcal W_s\right|,\end{equation}
By \eqref{hatRsmall}, it follows that
\begin{equation}\nonumber
\widehat \Phi\left(\zeta_{k_i};(w_j(k_i),\alpha_j(k_i),\beta_j(k_i))_{j=1}^{\mu}\right)=\left(I+\mathcal O\left(\frac{\epsilon_{k_i}}{(|\zeta_{k_i}|+1)}\right)\right)\widehat \Phi\left(\zeta_{k_i};(\mathcal W_j,\mathcal A_j,\mathcal B_j)_{j=1}^{\tau}\right),
\end{equation}
 as $k_i\to \infty$. By condition (c) for the RH problem for $ \Phi\left(\zeta;(\mathcal W_j,\mathcal A_j,\mathcal B_j)_{j=1}^\tau\right)$, and the definition of $\widehat \Phi$, it follows that 
\begin{multline}
F_\Phi\left(\zeta_{k_i};w_1(k_i),\dots,w_\mu(k_i)\right)
=\left(I+\mathcal O\left(\frac{\epsilon_{k_i}}{(|\zeta_{k_i}|+1)}\right)\right)
\Phi\left(\zeta_{k_i};(\mathcal W_j,\mathcal A_j,\mathcal B_j)_{j=1}^{\tau}\right)e^{\frac{\zeta_{k_i}}{2}\sigma_3}\\ \prod_{j=1}^\tau (\zeta_{k_i}-i\mathcal W_j)^{\mathcal B_j\sigma_3} \times
\exp\left[\pi i(-\mathcal B_j+\mathcal A_j)\chi_{\mathcal W_j}(\zeta_{k_i})\sigma_3\right]r(\zeta_{k_i}),\label{eqn101}
\end{multline}
where $r(\zeta)$ is given by
\begin{equation}\nonumber
r(\zeta)=\frac{\prod_{j=1}^\mu (\zeta-iw_j(k_i))^{\beta_j\sigma_3}}{\prod_{j=1}^\tau (\zeta-i\mathcal W_j)^{\mathcal B_j\sigma_3}},\end{equation}
where the branch cuts of $r$ are a subset of $[-iu/2,iu/2]$ and $r(\zeta)\to I$ as $\zeta \to \infty$. By condition (c) of the RH problem for $\Phi\left(\zeta_{k_i};(\mathcal W_j,\mathcal A_j,\mathcal B_j)_{j=1}^{\tau}\right)$, and the fact that $r(\zeta)=\mathcal O(1/\zeta)$ uniformly in $k_i$  as $\zeta \to \infty$, it follows that
\begin{equation}\nonumber
F_\Phi\left(\zeta_{k_i};w_1(k_i),\dots,w_\mu(k_i)\right)=\mathcal O(\zeta_{k_i}^{-1}),
\end{equation}
as $k_i \to \infty$. Thus the left hand side of \eqref{Contra} is bounded as $k_i\to \infty$, which is a contradiction, concluding the proof of Lemma \ref{Lemmacont} (a).

 To prove (b), we note that if $|w_1|,\dots,|w_m|<\epsilon$ and $\epsilon \to 0$, then similarly to \eqref{eqn101} we have
\begin{multline}\nonumber
\Phi(\zeta)e^{\frac{\zeta }{2}\sigma_3}\prod_{j=1}^\mu (\zeta-iw_j)^{\beta_j\sigma_3}\exp\left[\pi i(-\beta_j+\alpha_j)\chi_{w_j}(\zeta)\sigma_3\right]
\\=\left(I+\mathcal O\left(\frac{\epsilon}{(|\zeta|+1)}\right)\right)
\Phi\left(\zeta;0,\mathcal A,\mathcal B\right)e^{\frac{\zeta}{2}\sigma_3} \zeta^{\mathcal B\sigma_3}
\exp\left[\pi i(-\mathcal B+\mathcal A)\chi_{0}(\zeta)\sigma_3\right]r_0(\zeta),
\end{multline}
as $\zeta \to \infty$ and $\epsilon \to 0$, where
 $r_0(\zeta)$ is given by
\begin{equation}\nonumber
r_0(\zeta)=\frac{\prod_{j=1}^\mu (\zeta-iw_j(k_i))^{\beta_j\sigma_3}}{\zeta^{\mathcal B_1\sigma_3}},\end{equation}
where the branch cuts of $r_0$ are a subset of $[-iu/2,iu/2]$ and $r_0(\zeta)\to I$ as $\zeta \to \infty$. 
Then part (b) of the lemma follows by \eqref{1singPhi} and noting that $r_0(\zeta)=I+\mathcal O \left(\frac{\epsilon}{|\zeta|}\right)$ as $\zeta \to \infty$ and $\epsilon \to 0$.

\section{Asymptotics of the orthogonal polynomials}
\label{SectionAsym}

Define $Y=Y(z)$ in terms of the orthogonal polynomials:
\begin{equation}\label{Soln Y}
Y(z)=\begin{pmatrix}
\chi_n^{-1}\psi_n(z)&\chi_n^{-1} \int_{\mathcal C} \frac{\psi_n(\lambda)}{\lambda-z}\frac{f(\lambda)d\lambda}{2\pi i \lambda^n} \\
-\chi_{n-1}z^{n-1} \overline{\psi}_{n-1}(z^{-1})& -\chi _{n-1} \int_{\mathcal C} \frac{\overline{\psi_{n-1}(\lambda)}}{\lambda-z} \frac{f(\lambda)d\lambda}{2\pi i\lambda }
\end{pmatrix},
\end{equation}
with the integration taken in counter-clockwise direction on the unit circle $\mathcal C$, and where $\overline \psi_{n-1}(z)=\overline{\psi_{n-1}(\overline z)}$. 
The function $Y$ uniquely solves the following Riemann--Hilbert  Problem
\begin{itemize}
\item[(a)] $Y:\mathbb C \setminus \mathcal C \to \mathbb C^{2\times 2}$ is analytic;
\item[(b)]$Y_+(z)=Y_-(z)\begin{pmatrix} 1&f(z)z^{-n}\\0&1 \end{pmatrix}$ for $|z|=1$, \, $\arg z\neq t_1,t_2,\dots,t_m$;
\item[(c)] $Y(z)=(I+\mathcal{O}(1/z))\begin{pmatrix}z^n&0\\0&z^{-n}\end{pmatrix}$ as $z\to \infty$.
\end{itemize}
That $Y$ defined in \eqref{Soln Y} solves the RH problem for $Y$ is easily verified, and is a result due to Baik, Deift, Johansson \cite{BDJ}, who were inspired by a 
similar observation by  Fokas, Its, Kitaev \cite{FIK} concerning orthogonal polynomials on the real line.
It is immediate that 
\begin{equation}\label{chinY}
\chi_{n-1}^2=-Y_{21}(0).
\end{equation}
We rely on the Deift-Zhou \cite{DeiftZhou} steepest descent analysis for RH problems to obtain the asymptotics of $Y(0)$ as $n\to \infty$. See  e.g. \cite{Deift} for an introduction to analysis of RH problems.

The  Szeg\H{o} function $\mathcal D(z)=\exp\frac{1}{2\pi i }\int_C \frac{\log f(s)}{s-z}ds$ plays an important role. Define

\begin{equation}
\label{defDinDout}
\begin{aligned}
\mathcal D_{{\rm in}}(z)&=e^{\sum_{j=0}^\infty V_jz^j}\prod_{j=1}^m\left(\frac{z-e^{it_j}}{e^{it_j}e^{\pi i}}\right)^{\alpha_j+\beta_j}&&
\textrm{for $z\in \mathbb C \setminus \{z:\arg z=\arg t_j,|z|\geq 1\}$,}
\\
\mathcal D_{{\rm out}}(z)&=e^{-\sum_{j=-\infty}^{-1} V_jz^j}\prod_{j=1}^m \left(\frac{z-e^{it_j}}{z}\right)^{-\alpha_j+\beta_j}&&
\textrm{for $z\in \mathbb C \setminus\left(\{0\}\cup \{z:\arg z=\arg t_j,|z|\leq 1\}\right)$,}
\end{aligned}
\end{equation} 
analytic on $\mathbb C\setminus \{z:\arg z=\arg t_j\}$. In \cite{DIK}, it was noted that for $|z|<1$, we have $\mathcal D(z)=\mathcal D_{{\rm in}}(z)$ and for $z>1$ we have $\mathcal D(z)=\mathcal D_{{\rm out}}(z)$.

Furthermore, 
\begin{equation}\label{defComplexf}
f(z)=\mathcal D_{{\rm in}}(z)
\mathcal D_{{\rm out}}(z)^{-1}\end{equation}for $z\in \mathcal C\setminus \left(\cup_{j=1}^me^{it_j}\right) $, and we extend the definition of $f$ by letting $f$ be defined by \eqref{defComplexf} on $\mathbb C\setminus \left(\{0\}\cup  \{z:\arg z=\arg t_j\}\right)$. It follows that on $\{z:\arg z=\arg t_j\}$,
\begin{equation}\label{jumpsf}
f_+(z)=f_-(z)\begin{cases}
e^{2\pi i(\alpha_j-\beta_j)}& \textrm{for }0<|z|<1\\
e^{-2\pi i (\alpha_j+\beta_j)}&\textrm{for }|z|>1,
\end{cases}
\end{equation}
with the orientation taken away from $0$ and toward $\infty$.

\subsection{Transformation of the RH problem for $Y$, and opening of the lens}
Define 
\begin{equation}\label{defT}
T(z)=\begin{cases}
Y(z)&{\rm for}\,\, |z|<1,\\
Y(z)z^{-n\sigma_3}& {\rm for}\,\, |z|>1.
\end{cases}
\end{equation}

Let $\widehat U>0$ be such that the asymptotics of Lemma \ref{Lemmacont} (a) hold for $\zeta>\widehat U/3$, for any $\mu=1,2,\dots,m$.
Recall the notation from Section \ref{SecMethod}. 
Assume that $t_1,\dots,t_m$ satisfies condition $(u,\widehat U,n)$. In this section we denote $\textbf{Cl}_j=
\textbf{Cl}_j(u,\widehat U,n)$. 
Denote the number of points in each set $\textbf{Cl}_j$ by $\mu_j$ for $j=1,2,\dots,r$, and let 
\begin{equation}\label{deftj}
\widehat t_j= \frac{1}{\mu_j}\sum_{x \in \textbf{Cl}_j} x.\end{equation}
We denote the elements of $\textbf{Cl}_j=\{t^{(j)}_1,\dots,t^{(j)}_{\mu_j}\}$ for $j=1,2,\dots,r$, and denote the degree of the singularity of the Toeplitz determinant at $t^{(j)}_i$ by $\alpha_i^{(j)},\beta_i^{(j)}$, and order the parameters so that $t^{(j)}_1>\dots>t^{(j)}_{\mu_j}$. In this way we have a natural partition
\begin{equation}\label{clusternot}
\{(t_j,\alpha_j,\beta_j)\}_{j=1}^m=\cup_{j=1}^r\left\{\left(t_i^{(j)},\alpha_i^{(j)},\beta_i^{(j)}\right)\right\}_{i=1}^{\mu_j}.
\end{equation}

It is clear that the parameters which we label $\alpha_i^{(j)}, \beta_i^{(j)}$ in this section are not in general in direct correspondence with the parameters of the same notation in Section \ref{ModelProblem}.

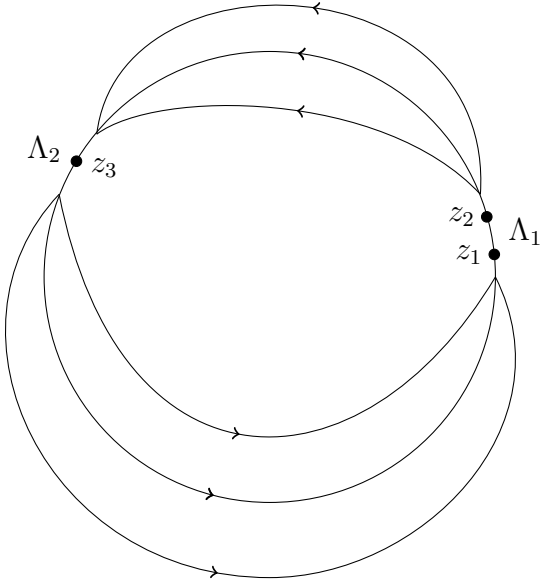
\begin{figure}
\begin{tikzpicture}

\draw  
        [decoration={markings, mark=at position 0.23 with {\arrow[thick]{>}}},
        postaction={decorate}]
        [decoration={markings, mark=at position 0.71 with {\arrow[thick]{>}}},
        postaction={decorate}]
        (3,3) circle (3cm);
	
\node at (6.4,3.6) 
{$\Lambda_1$};
\node at (0,4.7) {$\Lambda_2$};

\draw [decoration={markings, mark=at position 0.5 with {\arrow[thick]{>}}},
        postaction={decorate}]
        (0.2,4.1) .. controls (-1.5,2.3) and (0,-1) .. (3,-1) .. controls (5,-1) and (7,1)  .. (6,3);
\draw [decoration={markings, mark=at position 0.5 with {\arrow[thick]{>}}},
        postaction={decorate}]
        (0.2,4.1) .. controls (1,0) and (4.2,0)   .. (6,3);

\draw [decoration={markings, mark=at position 0.5 with {\arrow[thick]{>}}},
        postaction={decorate}]
        (5.8,4.1) .. controls (6,7.2) and (1,7.4)  .. (0.7,4.9);
\draw [decoration={markings, mark=at position 0.5 with {\arrow[thick]{>}}},
        postaction={decorate}]
        (5.8,4.1) .. controls (4.5,5.5) and (1.5,5.5)   .. (0.7,4.9);

\draw [fill=black] (5.985,3.3) circle[radius=0.07]; 
\draw [fill=black] (5.885,3.8) circle[radius=0.07]; 
\draw [fill=black] (0.43,4.54) circle[radius=0.07]; 

\node at (5.65,3.3) {$z_1$};
\node at (5.55,3.8) {$z_2$};
\node at (0.8,4.45) {$z_3$};

\end{tikzpicture}
\caption{
Opening of lenses and jump contour $\Sigma_S$ in the case of 3 singularities $z_1,z_2,z_3$ partitioned into two clusters.}\label{Lenses}
\end{figure}

We let
\begin{equation}\nonumber
\Lambda_j=\left\{z:|z|=1, \,\, -u\leq n\left(\arg (z)-\widehat t_j\right) \leq u\right\},
\end{equation}
so that $\{e^{it}:t\in \textbf{Cl}_j\}\subset \Lambda_j$ and so that $e^{it_1^{(j)}}, e^{it_{\mu_j}^{(j)}}$ are not the endpoints of the arc $\Lambda_j$,  and define $\Lambda=\cup_j \Lambda_j$. We open a lens around each arc comprising $ \mathcal C \setminus \Lambda$, where $\mathcal C$ is the unit circle, as in Figure \ref{Lenses}.

Let 
\begin{equation}\label{defS}
S(z)=\begin{cases}
T(z)&\textrm{outside the lens,}\\ 
T(z)\begin{pmatrix}
1&0\\z^{-n}f(z)^{-1}&1\end{pmatrix}& \textrm{inside the lenses and outside the unit disc,}\\
T(z)\begin{pmatrix}
1&0\\-z^{n}f(z)^{-1}&1\end{pmatrix}& \textrm{inside the lenses and inside the unit disc.}
\end{cases}
\end{equation}

By noting that
\begin{equation}\nonumber
T_+(z)=T_-(z)\begin{pmatrix}
z^n&f(z)\\0&z^{-n}
\end{pmatrix},
\end{equation} and by using the factorisation
\begin{equation}\nonumber
\begin{pmatrix}
z^n&f(z)\\0&z^{-n}
\end{pmatrix}
=\begin{pmatrix}
1&0\\z^{-n}f(z)^{-1}&1
\end{pmatrix}
\begin{pmatrix}
0&f(z)\\-f(z)^{-1}&0
\end{pmatrix}
\begin{pmatrix}
1&0\\z^nf(z)^{-1}&1
\end{pmatrix},
\end{equation}
it is easily verified that $S$ uniquely solves the following RH problem.
\subsubsection{RH problem for $S$}
\begin{itemize}
\item[(a)] $S$ is analytic on $\mathbb C\setminus \Sigma_S$, where $\Sigma_S$ is the union of the unit circle and the contours of the lenses.
\item[(b)] $S$ has the following jumps on $\Sigma_S$:
\begin{equation}\nonumber
\begin{aligned}
S_+(z)&=S_-\begin{pmatrix}
1&0		\\ z^{-n}f(z)^{-1}&1
\end{pmatrix}&& \textrm{on the contours of the lenses, $|z|>1$,}\\
S_+(z)&=S_-\begin{pmatrix}
1&0		\\ z^{n}f(z)^{-1}&1
\end{pmatrix}&& \textrm{on the contours of the lenses, $|z|<1$,}\\
S_+(z)&=S_-(z)\begin{pmatrix}
0&f(z)\\-f(z)^{-1}&0
\end{pmatrix}&&\textrm{for } z\in \mathcal C\setminus \Lambda,\\
S_+(z)&=S_-(z)
\begin{pmatrix}
z^n&f(z)\\0&z^{-n}
\end{pmatrix}&& \textrm{for }z\in \Lambda.
\end{aligned}
\end{equation}
\item[(c)] As $z\to \infty$,
\begin{equation}\nonumber
S(z)=I+\mathcal O(1/z).
\end{equation}
\item[(d)] As $z\to e^{it_j}$, $j=1,\dots,m$, in the region outside the lens,
\begin{equation}\nonumber
S(z)=\mathcal O(|\log |z-e^{it_j}||).
\end{equation}
\end{itemize}

\subsection{Main parametrix}
We define $M$ by
\begin{equation}
M(z)=\label{def:M}
\begin{cases}
\begin{pmatrix}
0&1\\-1&0
\end{pmatrix}
\mathcal D_{{\rm in}}(z)^{-\sigma_3}&{\rm for}\, |z|<1\\
\mathcal D_{{\rm out}}(z)^{\sigma_3}&{\rm for}\, |z|>1.
\end{cases}
\end{equation}
Then $M$ is analytic for $z\in \mathbb C\setminus \mathcal C$, by \eqref{defComplexf}
\begin{equation}\nonumber
M_+(z)=M_-(z)\begin{pmatrix}
0&f(z)\\-f(z)^{-1}&0
\end{pmatrix},
\end{equation}
and $M(z)=I+\mathcal O(z^{-1})$ as $z\to \infty$. Thus $M S^{-1}$ solves an RH problem with jumps that converge pointwise to $I$ as $n\to \infty$ except on the shrinking contour $\Lambda$, and $(M S^{-1})(z)\to I$ as $z\to \infty$, so we take $M$ to be our main parametrix.

\subsection{Local parametrix}
\label{SectionLocal}
We define open sets $\mathcal U_1,\dots,\mathcal U_r$ containing each cluster $\textbf{Cl}_1,\dots,\textbf{Cl}_r$ respectively by
\begin{equation}\nonumber
\mathcal U_j=\left\{z:\left|\log z-i\widehat t_j \right|<\widehat u_n/3n\right\},
\end{equation}
where we recall $\widehat t_j$ from \eqref{deftj} and $\widehat u_n$ from \eqref{hatun}. Let 
\begin{equation}\label{defzetaj}
\zeta_j(z)=n\left(\log z-i\widehat t_j\right),
\end{equation}
for $z\in \mathcal U_j$. 

Recall the notation \eqref{clusternot}. We define 
\begin{equation}
w^{(j)}_\nu=-i\zeta\left(e^{it_\nu^{(j)}}\right)=n\left(t_\nu^{(j)}-\widehat t_j\right) \end{equation}
  for $\nu=1,\dots,\mu_j$.
Then $\zeta_j$ is a conformal map on $\mathcal U_j$ mapping $\Lambda_j$ to $[-iu,iu]$. For $z\in \partial \mathcal U_j$,
\begin{equation}\nonumber
|\zeta_j(z)|=\widehat u_n/3.
\end{equation}
Recall the model RH problem $\Phi$ from Section \ref{ModelProblem}.
On $\mathcal U_j$, we define 
\begin{equation}\nonumber
P_j(z)=E_j(z)\Phi\left(\zeta_j(z);w_1^{(j)},\dots,w_{\mu_j}^{(j)}\right)
z^{\frac{n}{2}\sigma_3}W(z),\qquad W(z)=\begin{cases}
-\sigma_3 f(z)^{-\frac{1}{2}\sigma_3}&{\rm for}\, |z|<1,\\
 \begin{pmatrix}0&1\\1&0\end{pmatrix}f(z)^{-\frac{1}{2}\sigma_3}&{\rm for}\, |z|>1,
\end{cases}
\end{equation}
where 
\begin{equation}\nonumber
E_j(z)=M(z)W(z)^{-1}\prod_{\nu=1}^{\mu_j} \left(\zeta_j(z)-iw^{(j)}_\nu\right)^{\beta_\nu^{(j)}\sigma_3}\exp\left[\pi i\left(\alpha_\nu^{(j)}-\beta_\nu^{(j)}\right)\chi_{w_\nu^{(j)}}\left(\zeta_j(z)\right)\sigma_3\right]e^{-\frac{in\widehat t_j}{2}\sigma_3},
\end{equation}
recalling that $\chi_w(\zeta)$ was defined in condition (c) of the RH problem for $\Phi$,
and the branches of $\zeta_j(z)-iw_\nu^{(j)}$ are such that $\arg\left(\zeta_j(z)-iw_\nu^{(j)}\right)\in (0,2\pi)$.
By the jumps of $f$ in \eqref{jumpsf}, the definition of $M$ in \eqref{def:M}, and the definition of $W$, it follows that $E_j$ has no jumps on $\mathcal U_j$. By the definitions of $\mathcal D_{{\rm in}}$ and $\mathcal D_{{\rm out}}$ in \eqref{defDinDout}, of $f$ in \eqref{defComplexf}, and of $M$ in \eqref{def:M}, it is easily seen that $M(z)W(z)^{-1}$ is bounded on $\overline{\mathcal U_j}$. Thus $E_j$ is analytic on $\mathcal U_j$, and uniformly bounded on $\partial \mathcal U_j$ as $n\to \infty$.

By the jumps of $f$ in \eqref{jumpsf} and condition (b) for the RH problem for $\Phi$, it follows that $P_j$ and $S$ have the same jumps on $\mathcal U_j$. By condition (c) for the RH problem for $\Phi$, and the boundedness of $E_j(z)$ on $\partial \mathcal U_j$,
\begin{equation}\label{PM-1}
P_j(z)M^{-1}(z)=I+\mathcal O\left(\widehat u_n^{-1}\right),
\end{equation}
as $n\to \infty$, uniformly for $z\in \mathcal U_j$, and by \eqref{Phiinfunif} the error term is also uniform for $t_1,\dots,t_m$ satisfying condition $(u,\widehat U,n)$.

\subsection{Small norm matrix}
Let $ R$ be given by
\begin{equation}\label{def:R}
 R(z)=\begin{cases}
 S(z)M(z)^{-1} &\textrm{for } z\in \mathbb C \setminus  \left(\cup_{j=1}^r  \mathcal U_j\right), \\
 S(z)P_j(z)^{-1} &\textrm{for } z\in U_j,\, j=1,\dots,r.
\end{cases}
\end{equation}
 $R$ satsisfies the following RH problem.
\subsubsection*{RH problem for $R$}
\begin{itemize}
\item[(a)]
 $R$ is analytic on $\mathbb C\setminus\Sigma_R$, where $\Sigma_R$ is the union of the edges of the lenses and $ \cup_{j=1}^r\partial \mathcal U_j$.
\item[(b)]On $\Sigma_R$,
 \begin{equation}\nonumber
 R_+(z)=R_-(z)(I+\Delta(z)),
\end{equation}  
with orientiation taken clockwise, and by \eqref{PM-1} and condition (b) for the RH problem for $S$ we have 
\begin{equation}\label{OrderDelta}
\Delta(z)=\begin{cases}
\mathcal O(\widehat u_n^{-1}) &
\textrm{ for $z\in \partial \mathcal U_j$ and $j=1,\dots,r$,}\\
\mathcal O(|z|^n)&\textrm{ for $z$ on the edges of the lenses and $|z|<1$,}
\\
\mathcal O(|z|^{-n})&\textrm{ for $z$ on the edges of the lenses and $|z|>1$,}\end{cases}
\end{equation}
as $n\to \infty$, uniformly for $t_1,\dots,t_m$ satisfying condition $(u,\widehat U, n)$.
\item[(c)]
 $R(z)=I+\mathcal O(z^{-1})$ as $z\to \infty$.
\end{itemize}

\begin{lemma}Let $\alpha_j\geq0$ and $\Re \beta_j=0$ for $j=1,\dots,m$. 
$R(z)=I+\mathcal O(\widehat u_n^{-1})$, as $n\to \infty$, uniformly for $z\in  \mathbb C$ and $t_1,\dots, t_m$ satisfying condition $(u,\widetilde U, n)$, for some $\widetilde U>0$.
\end{lemma}
\begin{proof}
Small--norm analysis of RH--problems with fixed contours is standard material, see e.g. \cite{Deift}, but for RH--problems with shrinking contours the theory is less developed. In the following, we follow \cite{CFLW}, where a slightly more detailed description may be found for a similar problem.

It is easily verified that
\begin{equation}\label{Standard1}
R(z)=I+\frac{1}{2\pi i}\int_{\Sigma_R}\frac{R_-(s)\Delta(s)}{s-z}ds.
\end{equation}
Consider
\begin{equation}\nonumber
R_{\max}=\sup_{z\in \mathbb C, j,k\in \{1,2\}}|R_{j,k}(z)|,\end{equation}
and assume this maximum is acheived at $z_{\max}\in \mathbb C\cup \{\infty\}$ (or that $R_+$ or $R_-$ acheives this supremum at $z_{\rm max}$). We piecewise analytically continue $R_-$ and $\Delta$ to strips of width of order $2c\widehat u_n/n$ containing $\Sigma_R$, for some fixed but sufficiently small $c>0$. On these strips the bounds on $\Delta$ from \eqref{OrderDelta} still hold. Furthermore, on these strips $R$ is either equal to $R_-$ or  $R_-(I+\Delta)$, either way it follows by \eqref{OrderDelta} that
\begin{equation}\label{R-max}
\max_{j,k\in \{1,2\}}|R_-(z)|_{j,k}\leq 2R_{\max},\end{equation}
for $n$ sufficiently large, for all $z$ in the strips.  By deforming the contour of integration $\Sigma_R$, but keeping it in the strips, we may assume that $z_{\max}$ is of distance greater than $c\widehat u_n/n$ from $\Sigma_R$. Crucially, \eqref{R-max} still holds on this deformed contour, and combined with \eqref{Standard1}, it follows that
\begin{equation}\nonumber
R_{\max}\leq 1+R_{\max}\max_{j,k=\{1,2\}}\left|\int_{\Sigma_R}\left|\frac{\Delta(s)}{s-z_{\max}}\right|ds\right|_{j,k}, 
\end{equation}
where we now assume that $z_{\max}$ is of distance greater than $c\widehat u_n/n$ from $\Sigma_R$.
Thus 
\begin{equation}\nonumber
R_{\max}\leq  \frac{1}{1-\max_{j,k=\{1,2\}}\left|\int_{\Sigma_R}\left|\frac{\Delta(s)}{s-z_{\max}}\right|ds\right|_{j,k}}.
\end{equation}
By the fact that $z_{\max}$ is of at least distance $c\widehat u_n/n$ from $\Sigma_R$ for some $c>0$, by \eqref{OrderDelta}, and by the fact that $\partial \mathcal U_j$ is of length of order $\widehat u_n/n$, it follows that
\begin{equation}\label{intR1a}
\left|\sum_{j=1}^r \int_{\partial U_j}\left|\frac{\Delta(s)}{s-z_{\textrm{max}}}\right| ds\right| =\mathcal O\left( \frac{1}{\widehat u_n}\right),\end{equation}
as $n\to \infty$, uniformly for $t_1,\dots,t_m$ satisfying condition $(u,\widetilde U, n)$. Let $\Sigma_{\textrm{Edge}}^{\textrm{out}}$ denote the edges of the lenses in the exterior of the unit disc in the complex plane. Then
\begin{equation}\label{intR2a}
\left|\int_{\Sigma_{\textrm{Edge}}^{\textrm{out}}}\left|\frac{\Delta(s)}{s-z_{\textrm{max}}}\right| ds\right| 
=\mathcal O\left(\frac{n}{\widehat u_n}\left|\int_{\Sigma_{\textrm{Edge}}^{\textrm{out}}}\left|s^{-n}\right| ds\right| \right)
=\mathcal O\left(e^{-\widehat u_n}\right),\end{equation}
for $\widehat u_n>\widetilde U$ and some $\widetilde U>0$.
It follows that
\begin{equation}\label{Rmax}
R_{\max}<2.
\end{equation}
Now consider
\begin{equation}\nonumber
(R-I)_{\max}=\sum_{z\in \mathbb C, j,k\in \{1,2\}}|R(z)-I|_{j,k},\end{equation}
and assume this supremum is acheived at $z_{\max,2}\in \mathbb C\cup \{\infty\}$. By deforming the contour of integration, we may assume that $z_{\max,2}$ is of distance greater than $c\widehat u_n/n$ from $\Sigma_R$, for some constant $c>0$. Thus, by \eqref{Rmax}, \eqref{R-max}, \eqref{Standard1}, it follows that
\begin{equation}\nonumber
(R-I)_{\max}\leq \frac{2}{\pi}\max_{j,k=\{1,2\}}\left|\int_{\Sigma_R}\frac{\Delta(s)}{s-z_{\max,2}}ds\right|_{j,k}.
\end{equation}
The lemma follows upon integration, by similar arguments to \eqref{intR1a} and \eqref{intR2a}.
\end{proof}

 \begin{lemma}
 \label{lemmaR}
 Let $\alpha_j\geq 0$ and $\Re \beta_j=0$. Then the following two statements hold.
 
 \begin{itemize}
\item[(a)] As $n \to\infty$,
\begin{equation}\nonumber
R(0)=I+\mathcal O(1/n),
\end{equation}
uniformly for $t_1,\dots, t_m$ satisfying condition $(u,\widetilde  U,n)$. 
\item[(b)] As $n\to \infty$
\begin{equation}\nonumber
R(0)=I+\sum_{j=1}^r\int_{\partial U_j}\frac{\Delta_1(s)}{s}\frac{ds}{2\pi i }+\mathcal O\left(\frac{1}{n\widehat u_n}\right),
\end{equation}
 uniformly for $t_1,\dots,t_m$ satisfying condition $(u,\widetilde U,n)$, where
\begin{equation}\nonumber
\Delta_{1,22}(z)=\frac{\Phi_{1,11}\left(\mu_j;\left(w_\nu^{(j)},\alpha_{\nu}^{(j)},\beta_\nu^{(j)}\right)_{\nu=1}^{\mu_j}\right)}{\zeta_j(z)},
\end{equation}
for $z\in \partial U_j$.
\end{itemize}
\end{lemma}

\begin{proof}
We evaluate \eqref{Standard1} as $n\to \infty$. The integration contour $\Sigma_R$ partitions naturally into two parts, $\cup_{j=1}^r \partial \mathcal U_j$ and the edges of the lenses $\Sigma_{{\rm edge}}$.  Denote the edges of the lenses on the inside of the unit disc by $\Sigma_{{\rm Edge }}^{{\rm in}}$. Then
\begin{equation}\label{Standard2}
\left|\int_{\Sigma_{{\rm Edge }}^{{\rm in}}}\frac{R_-(s)\Delta(s)}{s}ds\right|_{jk}\leq C_1\left|\int_{\Sigma_{{\rm Edge }}^{{\rm in}}}\frac{s^n}{s}ds\right|
\leq C_2 \frac{\left(1-C_3\frac{\widehat u_n}{n}\right)^n}{n}\leq \frac{C_4}{ne^{C_3\widehat u_n}},
\end{equation}
for some constants $C_i>0$, $i=1,\dots,4$, and sufficiently large $\widehat u_n$, for $j,k\in \{1,2\}$. A similar statement can be made for the edges of the lense outside the unit disc. 
We note that the length of the contour $\cup_{j=1}^r \partial \mathcal U_j$ is of order $\widehat u_n/n$ as $n\to \infty$, and so since $\Delta(z)=\mathcal O(1/\widehat u_n)$ for $z\in \cup_{j=1}^r \partial \mathcal U_j$, it follows that
\begin{equation}\label{Standard3}
\int_{\cup_{j=1}^r \partial U_j}\frac{R_-(s)\Delta(s)}{s-z}ds =\mathcal O(1/n),\end{equation}
as $n\to \infty$. Part (a) of the lemma follows from \eqref{Standard2}-\eqref{Standard3}.
By the definition of $R$ in \eqref{def:R} and condition (c) of the RH problem for $\Phi$, we have that $\Delta(z)=\Delta_1(z)+\mathcal O(\widehat u_n^{-2})$ as $n\to \infty$, where
\begin{equation}\nonumber
\Delta_1(z)=\frac{E_j(z)\Phi_1E_j^{-1}(z)}{\zeta_j(z)},
\end{equation} 
for $z\in \partial U_j$. By \eqref{Standard2} it follows that as $n\to \infty$,
\begin{equation}\label{R0}
R(0)=I+\sum_{j=1}^r \int_{\partial U_j} \frac{\Delta_1(s)}{s}\frac{ds}{2\pi i }+\mathcal O\left(n^{-2}\right)+\mathcal O\left(\frac{1}{n\widehat u_n}\right),
\end{equation}
where the orientation of the integral is clockwise. Since 
\begin{equation}\nonumber
E_j(z)=\begin{pmatrix}
0&1\\1&0
\end{pmatrix}g_{E,j}^{\sigma_3}(z),
\end{equation}
for some analytic function $g_{E,j}$, it follows that
\begin{equation}\nonumber
\Delta_{1,22}(z)=\Phi_{1,11}/\zeta_j(z),
\end{equation}
and thus we have proven part (b) of the lemma.

\end{proof}

\subsubsection{Proof of Proposition \ref{PropPoly}}
By \eqref{chinY}, and the definition of $T,S,R$ in \eqref{defT}, \eqref{defS}, \eqref{def:R}, it follows that
\begin{equation}\label{chiRM}
\chi_{n-1}^2=-(R(0)M(0))_{21}.
\end{equation}
By \eqref{def:M},
\begin{equation}\label{Mat0}
M(0)=\begin{pmatrix} 0&e^{V_0}\\-e^{-V_0}&0\end{pmatrix},\end{equation}
 and by Lemma \ref{lemmaR}, $R(0)=I+\mathcal O(n^{-1})$ as $n\to \infty$, and it follows that
\begin{equation}\nonumber
\left|\log \chi_n+V_0/2\right|=\mathcal O\left(n^{-1}\right),
\end{equation}
as $n\to \infty$, uniformly for all $t_1,\dots, t_m$ satisfying condition $(u,\widetilde U, n)$,
which proves Proposition \ref{PropPoly2}.

By  \eqref{Mat0}, \eqref{chiRM} and Lemma \ref{lemmaR} (b),
\begin{equation}\label{chi11}
\chi_{n-1}^2=e^{-V_0}\left(1+\sum_{j=1}^r\Phi_{1,11}\left(\mu_j;\left(w_\nu^{(j)},\alpha_{\nu}^{(j)},\beta_\nu^{(j)}\right)_{\nu=1}^{\mu_j}\right)\int_{\partial U_j}\frac{ds}{2\pi i s \zeta_j(s)}\right)+\mathcal O\left(\frac{1}{n\widehat u_n}\right),
\end{equation}
as $n\to \infty$, uniformly for $t_1,\dots,t_m$ satisfying condition $(u,\widetilde U,n)$, with the integral taken with clockwise orientation.
By  \eqref{Phi1limsing},
\begin{equation}
\Phi_{1,11}\left(\mu_j;\left(w_\nu^{(j)},\alpha_{\nu}^{(j)},\beta_\nu^{(j)}\right)_{\nu=1}^{\mu_j}\right)=\left(\sum_{\nu=1}^{\mu_j} \alpha_\nu^{(j)}\right)^2-\left(\sum_{\nu=1}^{\mu_j} \beta_\nu^{(j)}\right)^2 +\mathcal O\left(w_{\mu_j}^{(j)}-w_1^{(j)}\right),
\end{equation}
as $w_{\mu_j}^{(j)}-w_1^{(j)}\to 0$. By the definition of $\zeta_j$ in \eqref{defzetaj} and the fact 
 that the orientation of the integral is clockwise, it follows that
 \begin{equation}\label{chi13}
 \int_{\partial U_j} \frac{ds}{2\pi i s \zeta_j(s)}=-1/n.
 \end{equation}

Combining \eqref{chi11}-\eqref{chi13} proves Proposition \ref{PropPoly} (b).

\section*{Acknowledgements}
I am grateful to Igor Krasovsky for useful discussions and suggestions.
The author was supported by the G\"{o}ran
Gustafsson Foundation (UU/KTH) and by the Leverhulme Trust research programme grant
RPG-2018-260.

\section*{Appendix}
\subsubsection*{Proof of Proposition \ref{PropSolv}}
Define
\begin{equation}\nonumber
V(\zeta)=\Phi(\zeta)e^{\frac{\zeta}{2}\sigma_3}\prod_{j=1}^\mu (\zeta-iw_j)^{\beta_j\sigma_3}\exp\left[-\pi i (\beta_j-\alpha_j)\chi_{w_j}(\zeta)\sigma_3\right].
\end{equation}
Then $V$ is analytic on $\mathbb C\setminus \Gamma_V$, where $\Gamma_V= \cup_{j=0}^4\Gamma_j$. On $\Gamma_V\setminus \left(\pm iu \cup_{j=1}^\mu w_j\right)$, the jump matrix of $V$ factorizes into
\begin{equation}\nonumber
J_V=J_{V,-}^{-1}J_{V,+},
\end{equation}
where $J_{V,-}$ is upper triangular and piecewise analytic, and $J_{V,+}$ is lower triangular and piecewise analytic, and 
\begin{equation}\nonumber
|J_{V,+}(\zeta)-I|,\, |J_{V,-}(\zeta)-I|=\mathcal O(e^{-|\zeta|/2})\end{equation}
 as $\zeta \to \infty$. Furthermore, as $\zeta\to \infty$,
\begin{equation}\nonumber
V(\zeta)=I+\mathcal O(\zeta^{-1}).\end{equation}

 For RH problems of the form $V$, it is well known (see \cite{Van1} and \cite{Van2,Van3,Van4}) that $V$ has a unique solution if and only if the homogenous RH problem $V_{{\rm Hom}}$ has a unique solution, namely $0$, where $V_{{\rm Hom}}$ is analytic on $\mathbb C\setminus \Gamma_V$, has jumps $J_V$ on $\Gamma_V$, and satisfies
\begin{equation}\nonumber
V_{{\rm Hom}}(\zeta)=\mathcal O(\zeta^{-1}),
\end{equation}
as $\zeta \to \infty$. 

We will find it easier to work with $\widehat \Phi_{{\rm Hom}}$ which we define below. It is easily verified that if $\widehat \Phi_{{\rm Hom}}$ has the zero solution as its unique solution, then the same holds for $V_{{\rm Hom}}$.

We consider the homogenous RH problem for $\widehat \Phi$. Namely, we search for a function 
 $\widehat \Phi_{{\rm Hom}}$ satisfying conditions (a), (b), and (d) in the RH problem for $\widehat \Phi$, and as $\zeta \to \infty$,
   \begin{equation}\nonumber
  \widehat \Phi_{{\rm Hom}}(\zeta)e^{\frac{\zeta}{2}\sigma_3}=\mathcal O\left(\zeta^{-1}\right).
   \end{equation}
   We will prove that $\widehat \Phi_{{\rm Hom}}(\zeta)=0$ is the only function satisfying these conditions, and by the discussion above, it follows that the RH problems for $\widehat  \Phi$ and $\Phi$ have  unique solutions.
   Let $U(\zeta)=\widehat \Phi_{{\rm Hom}}(\zeta)e^{\frac{\zeta}{2}\sigma_3}$ and let
   \begin{equation}\nonumber
   W(\zeta)=U(\zeta)U^*(-\overline \zeta),
   \end{equation}
   where $^*$ denotes the conjugate transpose. Then $W$ is analytic on $\mathbb C \setminus (-i\infty,i\infty)$, and we take the orientation of $(-i\infty,i\infty)$ upwards.
We note that if $x\in \mathbb R$, then as $\zeta\to ix$ from the "+" side, it follows that $-\overline \zeta \to ix$ from the "-" side.
 Thus, by conditions (b) and (d) of the RH problem for $\Phi$ and the definitions of $\widehat \Phi$, $U$, $W$, as $\zeta \to iw_j$ from the $+$ side,
   \begin{equation}\nonumber 
   W(\zeta)=F_j(\zeta)(\zeta-iw_j)^{\alpha_j\sigma_3}\begin{pmatrix}
   \mathcal O\left(|\log(\zeta-iw_j)|^2\right)&    \mathcal O\left(|\log(\zeta-iw_j)|^2\right)\\
   \mathcal O\left(|\log(\zeta-iw_j)|^2\right)&0
   \end{pmatrix}
 \left( \overline{ -\overline{\zeta}-iw_j}\right)^{\alpha_j\sigma_3}F_j^*(-\overline \zeta).
   \end{equation}
   Thus $W_+(\zeta)$ is integrable for $\zeta\in(-i\infty,i\infty)$, and since $W(\zeta)=\mathcal O\left(|\zeta|^{-2}\right)$ as $\zeta\to \infty$, it follows from Cauchy's theorem that
   \begin{equation}\nonumber
    \int_{-i\infty}^{i\infty}W_+(\zeta)d\zeta=0.
\end{equation}    
For $\zeta\in(-i\infty,i\infty)$, we have $W_+(\zeta)=U_+(\zeta)U_-^*(\zeta)$, and so by the jump conditions for $\widehat \Phi$,
\begin{equation}\label{eqnU-}
\sum_{j=0}^\mu \int_{iw_{j+1}}^{iw_j}U_-(\zeta)e^{-\frac{\zeta}{2}\sigma_3}J_je^{\frac{\zeta}{2}\sigma_3}U_-^*(\zeta)d\zeta=   \int_{-i\infty}^{i\infty}W_+(\zeta)d\zeta=0,
\end{equation}
where again we take the convention $w_0=+\infty$ and $w_{\mu+1}=-\infty$. For purely imaginary $\zeta$, and with $\alpha_j\geq 0$ and $\Re \beta_j=0$, we have
\begin{equation}\nonumber
e^{-\frac{\zeta}{2}\sigma_3}J_je^{\frac{\zeta}{2}\sigma_3}+(e^{-\frac{\zeta}{2}\sigma_3}J_je^{\frac{\zeta}{2}\sigma_3})^*=\begin{pmatrix}
0&0\\0&2\exp\left[2\pi i\sum_{\nu=1}^j\beta_\nu\right]
\end{pmatrix},
\end{equation}
for $j=1,\dots,\mu$, and for $j=0$ the $22$ entry of the right hand side is taken to be $1$. Thus, summing \eqref{eqnU-} with its conjugate transpose, it follows that there is a strictly positive function $g$ such that
\begin{equation}\nonumber
\int_{-i\infty}^{i\infty}U_-(\zeta)\begin{pmatrix}
0&0\\0&g(\zeta)
\end{pmatrix}U_-^*(\zeta)d\zeta=0,
\end{equation}
and thus $U_{12,-}(\zeta),U_{22,-}(\zeta)=0$ for $\zeta \in (-i\infty,i\infty)$. From the jump conditions of $\widehat \Phi$, it follows that $U_{11,+}(\zeta),U_{21,+}(\zeta)=0$ for $\zeta \in (-i\infty,i\infty)$. From the identity theorem it follows that the first column of $U(\zeta)=0$ for $\Re \zeta<0$ and the second column of $U(\zeta)=0$ for $\Re \zeta>0$.

For $j=1,2$, let 
\begin{equation}\nonumber
g_j(\zeta)=\begin{cases}
U_{j2}(\zeta)&{\rm for}\, \Re\zeta<0\\
U_{j1}(\zeta)e^{-\zeta}&{\rm for}\, \Re\zeta>0,\, \Im \zeta>w_1\\
U_{j1}(\zeta)e^{-\zeta}\exp\left[-2\pi i \sum_{\nu=1}^\mu \alpha_\nu\right]&{\rm for}\, \Re \zeta>0,\, \Im \zeta<w_\mu.
\end{cases}
\end{equation}
By the definitions of $g_j$ and $U$, and condition (b) of the RH problem for $\widehat \Phi$, it follows that $g_j$ is analytic on $\mathbb C\setminus\{z:\Re z\geq 0,\, \Im z\in [w_\mu,w_1]\}$. Furthermore, if
\begin{equation}\nonumber
h_j(\zeta)=g_j(-(\zeta+u)^{3/2}),
\end{equation}
then $h_j$ is analytic and bounded for $\Re \zeta>0$, and $h_j(\zeta)=\mathcal O\left(e^{-|\zeta|/2}\right)$ as $\zeta \to \pm i\infty$. Thus it follows by Carlson's theorem (see e.g. \cite{Tit}), that $h_j(\zeta)=0$ for $\Re \zeta>0$, and by analytic continuation it follows that $g_j(\zeta)=0$ for $\zeta$ in the domain of $g_j$. It follows that $\widehat \Phi_{\textrm{Hom}}=0$.

\end{document}